\documentclass[oneside]{amsart}
\usepackage {CJKutf8}
\usepackage{amssymb}
\usepackage[abbrev]{amsrefs} 
\usepackage{hyperref}

\usepackage{enumitem}

\usepackage{todonotes}

\usepackage{float} 

\theoremstyle{plain}%
\newtheorem{theorem}{Theorem}
\newtheorem{proposition}[theorem]{Proposition}%
\newtheorem{lemma}[theorem]{Lemma}
\newtheorem{corollary}[theorem]{Corollary}

\theoremstyle{remark}%
\newtheorem{remark}{Remark}%

\theoremstyle{definition}%
\newtheorem{definition}{Definition}%

\numberwithin{equation}{section}
\numberwithin{theorem}{section}
\numberwithin{definition}{section}

\newcommand{\inn}[2]{\left\langle#1,\,#2\right\rangle}

\DeclareMathOperator{\supp}{supp}

\DeclareMathOperator{\dist}{dist}
\DeclareMathOperator{\sgn}{sgn}

\newcommand{\di}{\partial}

\DeclareMathOperator{\Tr}{Tr}

\newcommand{\br}[1]{\left\langle#1\right\rangle}
\newcommand{\si}{\sigma}

\newcommand{\g}{\gamma}
\newcommand{\al}{\alpha}

\newcommand{\Cb}{\mathbb{C}}

\newcommand{\Rb}{\mathbb{R}}

\newcommand{\Lc}{\mathcal{L}}

\newcommand{\F}{\mathcal{F}}
\newcommand{\loc}{\text{loc}}

\newcommand{\one}{\ensuremath{\mathbf{1}}}

\newcommand{\hf}{\ensuremath{\mathfrak{h}}}

\newcommand{\om}{\omega}
\newcommand{\Om}{\Omega}
\newcommand{\Ga}{\Gamma}

\renewcommand{\l}{\lambda} 

\newcommand{\abs}[1]{\ensuremath{\left\lvert#1\right\rvert}}
\newcommand{\norm}[1]{\ensuremath{\left\lVert#1\right\rVert}}
\newcommand{\sbr}[1]{\left[#1\right]}
\newcommand{\Set}[1]{\left\{#1\right\}}

\newcommand{\md}[6]{\ensuremath{
		\ifinner
		\tfrac{\partial{^{#2}}#1}{\partial{#3^{#4}}\partial{#5^{#6}}}
		\else
		\tfrac{\partial{^{#2}}#1}{\partial{#3^{#4}}\partial{#5^{#6}}}
		\fi
}}
\newcommand{\del}[1]{\left(#1\right)}
\newcommand{\thmref}[1]{Theorem~\ref{#1}}

\newcommand{\secref}[1]{Section~\ref{#1}}
\newcommand{\lemref}[1]{Lemma~\ref{#1}}
\newcommand{\propref}[1]{Proposition~\ref{#1}}
\newcommand{\remref}[1]{Remark~\ref{#1}}
\newcommand{\figref}[1]{Figure~\ref{#1}}
\newcommand{\corref}[1]{Corollary~\ref{#1}}

\numberwithin{equation}{section}

\usepackage[normalem]{ulem}
\usepackage{soul}
\setstcolor{blue}

\definecolor{green}{rgb}{0.0, 0.5, 0.5}

\definecolor{lgray}{gray}{0.9}
\definecolor{llgray}{gray}{0.95}
\definecolor{lllgray}{gray}{0.975}

\newcommand{\R}{\mathbb{R}}
\newcommand{\C}{\mathbb{C}}

\newcommand{\cB}{\mathcal{B}}
\newcommand{\cD}{\mathcal{D}}

\newcommand{\cF}{\mathcal{F}}
\newcommand{\cG}{\mathcal{G}}

\newcommand{\cL}{\mathcal{L}}
 
\newcommand{\cX}{\mathcal{X}}

\newcommand{\nc}{\newcommand}

\nc{\h}{\delta}
\nc{\G}{\Gamma}
\nc{\et}{\eta} 
\nc{\gam}{\gamma}
\nc{\ka}{\kappa}
\nc{\lam}{\lambda}
\nc{\Lam}{\Lambda}

\nc{\ta}{\tau}
\nc{\w}{\omega}
\nc{\io}{\iota}
\nc{\s}{\sigma}
\nc{\vphi}{\varphi}

\nc{\e}{\epsilon}

\nc{\ran}{\rangle}
\nc{\lan}{\langle}

\newcommand{\ra}{\rightarrow}

\renewcommand{\Im}{\mathrm{Im}} 

\newcommand{\tr}{\mathrm{Tr}}

\nc{\bfone}{{\bf 1}}
\nc{\1}{{\bf 1}}

\newcommand{\p}{\partial}

\nc{\dd}{\mathrm{d}}

\newcommand{\DETAILS}[1]{}

\usepackage{array}
\usepackage{multirow}

\newcommand{\astlo}{\hat\chi}

\DeclareMathOperator{\Ad}{ad}
\newcommand{\ad}[3]{\Ad^{#1}_{#2}(#3)}

\newcommand{\cp}{\mathrm{c}}

\newcommand{\Rem}{\mathrm{Rem}}
\DeclareMathOperator{\dG}{\mathrm{d}\Gamma}
\newcommand{\SD}{\mathrm{SD}}

\newcommand{\Norm}[1]{\norm{#1}}

\newcommand{\intn}{\mathrm{int}}

\newcommand{\ondel}[1]{\inn{\varphi}{{#1}\varphi}}
\newcommand{\phidel}[1]{\inn{\phi}{{#1}\phi}}
\newcommand{\Ondel}[1]{\inn{\Om}{{#1}\Om}}
\newcommand{\ordel}[1]{\inn{\varphi_r}{{#1}\varphi_r}}
\newcommand{\otdel}[1]{\inn{\varphi_t}{{#1}\varphi_t}}

\newcommand{\undel}[1]{\inn{\varphi}{{#1}\psi}}
\newcommand{\urdel}[1]{\inn{\varphi_r}{{#1}\psi_r}}

\newcommand{\wndel}[1]{\inn{\psi}{{#1}\psi}}
\newcommand{\wrdel}[1]{\inn{\psi_r}{{#1}\psi_r}}

\newcommand{\Undel}[1]{\inn{U_t^\xi\varphi}{{#1}U_t^\xi\varphi}}

\DeclareMathOperator{\WC}{WC}

\begin{document}
	\title[Propagation of information]{On propagation of information in quantum many-body systems
	}
	\author[I.~M.~Sigal]{Israel Michael Sigal}
	\address{Department of Mathematics, University of Toronto, Toronto, ON M5S 2E4, Canada }
	\email{im.sigal@utoronto.ca}
	
	\author[J.~Zhang]{Jingxuan Zhang
	}
	\address{Yau Mathematical Sciences Center,
		Tsinghua University,
		Haidian District,
		Beijing 100084, China }
	\address{Department of Mathematical Sciences, University of Copenhagen, Copenhagen 2100, Denmark}
	\email{jingxuan@tsinghua.edu.cn}

	\date{\today}
	\keywords{Maximal propagation speed; Lieb-Robinson bounds; quantum many-body  systems; quantum information; quantum light cones}
	\begin{abstract}
We prove bounds on the minimal time for quantum messaging, propagation/creation of correlations, and control of states  for \textit{general} lattice quantum many-body systems. The proofs are based on a maximal velocity bound, which states that the many-body evolution stays, up to small leaking probability tails, within a light cone of the support of the initial conditions. This estimate is used to prove the light-cone approximation of dynamics and Lieb-Robinson-type bound, which in turn yield the results above. Our conditions cover long-range interactions. The main results of this paper as well as some key steps of the proofs were first presented in \cite{LRSZ}.

	\end{abstract}

	\maketitle

	\section{Introduction}\label{sec:1}

The finite speed of propagation of particles and fields 
  is a fundamental law of nature. It provides powerful constraints in relativistic physics. It is remarkable that such constraints also effectively exist in non-relativistic quantum theory, the only quantum theory with a solid mathematical foundation and physical consistency. 
This was discovered by Lieb and Robinson (\cite{LR}) 50 years ago, 
 for quantum spin lattice systems,  in a form of space-time bounds   on the commutators of observables with disjoint space-time supports. 


About  40 years later, starting with the work of Hastings (\cite{H1}) on  the Lieb-Schultz-Mattis theorem 
and  followed by Nachtergaele and Sims (\cite{NachS}) on exponential decay of correlations in  condensed matter physics,  
Bravyi, Hastings and Verstraete (\cite{BHV}),  Bravyi and Hastings (\cite{BH}),    Eisert and Osborn  (\cite{EisOsb})  and Hastings (\cite{H2, H3}) on quantum messaging, correlation creation, scaling and  area laws for the entanglement entropy and belief propagation in Quantum Information Science (QIS), 
 it transpired that Lieb-Robinson bounds (LRBs) are among the very few effective and general tools that are available for analyzing quantum many-body systems.\footnote{It is clear in hindsight that Lieb-Robinson-type bounds on the propagation of information would play a central role in QIS.} 

In the last 15 years, following these works, a new active area of theoretical and mathematical physics 
 dealing with dynamics of quantum information 
  sprung to life.  A variety of improvements of the original LRB, e.g., extensions to long-range spin interactions, fermionic lattice gases and finally to bosonic systems,  have been achieved, and  their applications expanded and deepened to include, e.g.  
 the state transport (\cite{EpWh, FLS2}) and the error bounds on quantum simulation algorithms (see e.g.~\cite{TranEtAl1, TranEtAl3}) in QIS,  the equilibration (\cite{GogEis}) in condensed matter physics, thermodynamic limit of dynamics (\cite{NachOgS, NachSchlSSZ, NachVerZ}) in Statistical Mechanics and scrambling time \DETAILS{{\bf (check)} } in  Quantum Field Theory \cite{
 RS}.   
 See the survey papers \cite{KGE, NSY2} and brief reviews in \cite{FLS, FLS2}. 
   
Independently and using a different approach,  it was shown in \cite{SigSof2} that in Quantum Mechanics the ``essential support'' of the wave functions, i.e.~the support up to negligible probability tails, spreads  with finite speed.
The  result was improved in \cite{Skib, HeSk, APSS} and extended to the nonrelativistic QED 
  in \cite{BoFauSig} 
  and to condensed matter physics, i.e.~\DETAILS{from the zero density quantum systems}to systems with positive particle densities, in \cite{FLS, FLS2}. 



In this paper, we prove bounds   on the minimal time for quantum messaging, propagation/creation of correlations (scrambling time\footnote{{The scrambling time could be defined as the time an initially uncorrelated subsystem stays uncorrelated (i.e.~the time  needed to create correlations). }In particular, our results imply non-existence of fast scramling, c.f.~\cite{KS_he,CL}.}) and control of states,  for general  quantum many-body lattice systems.\footnote{From the condensed matter physics viewpoint, such systems arise in the standard tight-binding approximation, see e.g.~\cite{FeffLThWein} and, for rigorous results, \cite{AshMerm, Harr}. We consider them to avoid inessential technicalities in the proof of the approximation result, \thmref{thm2}.}  Fixing a lattice $\cL\subset \Rb^d,\,d\ge1$, such systems are described by the Hamiltonians
\footnote{For background on the second quantization and quantum many-body systems, see \cite{BrRo, GS}.}
\begin{align}\label{1.1} 
 H_\Lam &:=  \sum_{x,y \in \Lambda} h_{xy} a_x^*a_y + \frac12\sum_{\substack{x,y\in \Lam
		}}a^*_xa^*_yv_{xy}a_ya_x,
\end{align} 
{for subsets $\Lam$ of $\cL$}, acting on the bosonic Fock spaces\footnote{{See Appendix \ref{sec:Fock-sp} for the definitions and discussions of Fock spaces}.}  $\mathcal{F}_\Lam$ over {the $1$-particle Hilbert spaces} $\hf_\Lam:=\ell^2(\Lambda)$.
Here $a_x$ and $a^*_x$ are the annihilation and creation operators, respectively,
$h_{xy}$ is the operator kernel (matrix) of an $1$-particle Hamiltonian $h$ acting on $\hf$ and $v_{xy}$ is a $2$-particle pair potential. {We assume that $h$ is hermitian, i.e.~$h_{xy}=\overline{h_{yx}}$ and $v$ 
	is real-symmetric, i.e.~$v_{xy}=\overline{ v_{xy}}=v_{yx}$\DETAILS{(e.g. $v_{xy}=v(x-y)$ for some $v:\Lam\to\Rb$ with $v(x)=v(-x)$)}. }

{Our main results are given in Theorems \ref{thm1}--\ref{thm4} below. Our starting point is the maximal velocity bound (MVB), \thmref{thm1}, which states that the many-body evolution stays, up to small leaking  tails, within a light cone of the support of the initial conditions. We use the MVB to derive  Theorem \ref{thm2} on  the light-cone approximation of quantum dynamics, which, in turn,  yields the weak LRB,  \thmref{thm3}. The latter 
 establishes power-law decay of commutators of evolving observables and holds (uniformly) on \textit{a subset of localized states}.


  
  Theorems \ref{thm5}--\ref{thm:gap} 
  provide general constraints on  propagation/creation of correlation, quantum messaging, state control times, { and the relation between a spectral gap and the decay of correlations.} They are derived readily from Theorems \ref{thm2} and \ref{thm3}. 
  Theorem \ref{thm4} describes macroscopic particle transport. Its proof extends in an essential way  the proof of Theorem \ref{thm1}.
  
For pure states,  the correlation signifies  entanglement and Theorem \ref{thm5} gives bounds on the time for  propagation/creation of entanglement between different regions  
within a given spatial domain.

 To emphasize, 
our results yield the existence of a linear light cone for general lattice quantum many-body  systems, providing powerful constraints on the evolution of information for such systems. 

The bounds on the maximal speed of propagation are given in terms of the norm of the $1$-particle group velocity operator $i[h, x]$,  
where $h,\,x$ are the $1$-particle Hamiltonian entering \eqref{1.1} and the position observable, respectively.}



The Hamiltonians under consideration in this paper are characterized by two decay rates, one for $h_{xy}$ and the other, for $v_{xy}$. 
We assume that there exists $n\ge1$ s.th.
	\begin{align}\label{k-cond}
		\kappa_n:=&\sup _{x\in\Lam} \sum_{y\in\Lam} \abs{h_{xy}}\abs{x-y}^{n+1}<\infty,\\
\label{v-cond}
		\nu_n:=&\sup_{x\in\Lam}	\sum_{y\in\Lam} \abs{v_{xy}}\abs{x-y}^n<\infty.
	\end{align}
Moreover, in \thmref{thm1}, we do not use condition \eqref{v-cond} and allow $v_{xy}$ to be arbitrary (apart from the standing assumption $v_{xy}=\overline{v_{xy}}=v_{yx}$).
	The parameter $n\ge1$ in \eqref{k-cond}--\eqref{v-cond} determines the time-/ space-decay rate in various statements.

All our results hold for bosonic systems with long-range interactions\footnote{In the present context, $H$ is said to be short-/long-range if $h_{xy},\,v_{xy}$ decay exponentially/polynomially  in $\abs{x-y}$.}, 
say, 
$\abs{h_{xy}}\leq C (1+|x-y|)^{-\alpha}$ with $\alpha>d+n+1$ and similarly for $v$, which suffices for \eqref{k-cond}--\eqref{v-cond}. 
	Taking $n=1$, we see that,  for $d>1$ and $\alpha\in (d+2,2d+1)$, our result gives
	a linear light cone as defined in terms of the weak LRB \eqref{LRB}. On the other hand, fast state-transfer and entanglement-generation protocols \cite{KS1,KS2,EldredgeEtAl,TranEtal4,TranEtal5} show that 
	 linear light cones, defined in terms of the LRB, do not exist for $\al<2d+1$.
See \cite{TranEtAl1} for the phase diagram summarizing the situation for the Lieb-Robinson light cones and \cite{Bose,DefenuEtAl} for  reviews of the effect of the long-range interactions  on quantum many-body dynamics and, in particular, on the transmission of quantum information. 
Thus our bounds narrow the class of systems for which long-range interactions lead  to speed-up of the spreading of information\footnote{We are grateful to Marius Lemm for pointing this out to us.}. 

Our results can be extended  to Hamiltonians with time-dependent and few-body interactions and adapted to long-range fermionic systems.  The main results of this paper as well as some key steps of the proofs were first presented in \cite{LRSZ}
 
%
%

 The main results of this paper as well as some key steps of the proofs were first presented in \cite{LRSZ}.
 \subsection{Related results}\label{sec:rr}

 
Results similar to  Theorems \ref{thm1}--\ref{thm3}, \ref{thm6}--\ref{thm:qst}, and \ref{thm4}, but for the Bose-Hubbard model, were obtained in \cite{FLS, FLS2}. Our proofs of those theorems follow  the corresponding proofs in \cite{FLS,FLS2}. 
The rendition in the present paper is more geometric, which streamlines the derivations and makes the arguments more transparent. 
 	Furthermore, we view the results in \thmref{thm:qst} as bounds on quantum control time, rather than those on the time for state transfer as in \cite{EpWh,FLS2}.

	Earlier on, results similar to Thms~\ref{thm1} and \ref{thm2} and \ref{thm3} have previously been obtained in \cite{SHOE}, \cite{YL} and \cite{KS2}, respectively (the last two papers deal with the Bose-Hubbard model), for detailed comparisons, see  \remref{rem:lit} in Section \ref{sec:ex}. 	
Moreover, LRB for a special class of bosonic  lattice systems was proved in \cite{MR2472026}.

	The constraints imposed by the LRB on the propagation/creation of correlations were first discussed in \cite{BHV}, with rigorous results for fermionic systems given in \cite{NachOgS}. 	
	\DETAILS{Using the LRB, the authors of \cite{NachOgS} proved an upper bound  on the rate of propagation/creation of correlations between observables with separated spatial support for long-range interacting fermions, provided the initial state is a product state (uncorrelated).}

	 The relation between spectral gap and the decay of correlations (clustering) for fermionic systems was established in \cite{H0,H1,HastKoma,NachS}, with the sharpest results given in \cite{NachS}. 


 


As we were preparing the present paper for publication, a new preprint \cite{KVS} was posted with deep results related to those in Thms.~\ref{thm1}--\ref{thm2} for quantum many-
body Hamiltonians with nearest-neighbour interactions.
Assuming the initial state satisfies a uniform low density condition, the authors  of \cite{KVS}  proved the existence of the superlinear light cone $\abs{x}\sim t\log t$ (resp.~$\abs{x}\sim t^d\,\mathrm{polylog}\,t$, where $d$ is the dimension and $\mathrm{polylog}$ is the polylogarithmic function) for particle transport  (resp.~the light-cone approximation of observables), up to fast decaying leaking probability tails. 

\subsection{Organization of the paper}\label{sec:1.1}

In \secref{sec:setup}, we describe the dynamics generated by the Hamiltonian \eqref{1.1} and formulate the main results of this paper, Theorems \ref{thm1}-\ref{thm4}. 
 Their proofs are given in Sects.~\ref{sec:pfthm1}--\ref{sec:pfthm4}. 
 The technicalities are  deferred to the appendices, with Appendix \ref{sec:Fock-sp} containing some general facts about the Fock spaces. In Section \ref{sec:ex}, we comment on possible extensions of the main theorems.  

\subsection{Notation}
Throughout the paper, we fix the underlying lattice $\cL$, with  grid size $\ge1$, and the domain $\Lam\subset\cL$, and we do not display these in our notations, e.g., we write $H$, $\cF$, and $\hf$ for $H_\Lam$, $\cF_\Lam$, and $\hf_\Lam$. 
	
	We denote by $\cD(A)$ the domain of an operator $A$ and 
$\|\cdot\|$,   the norm of  operators on $\cF$ and sometimes on $\hf$.
For a bounded operator $A$ on $\hf$, we denote by $A_{xy},\,x,y\in \Lam$, the operator kernel (matrix) of $A$. 
We make no distinction in our notation between a function $f\in\hf$ and the associated multiplication operator $\psi(x)\mapsto f(x)\psi(x)$ on $\hf$. 

All quantities and equations we work with are dimensionless. 	{In particular, in our units, the Planck constant is set to $2\pi$ and speed of light, to one ($\hbar = 1$ and $c = 1$).  }

\section{Setup and main results}\label{sec:setup}

{For symmetric $h$ and $v$, the Hamiltonian $H$ in \eqref{1.1} is symmetric and therefore self-adjoint. To show the latter, one can} use the canonical commutator relations to show that the number operator 
\begin{equation}\label{1.4}N\equiv N_\Lam,\quad \text{where}\quad N_X := \sum_{x\in X}a^*_xa_x,\end{equation} commutes with $H$. Since {the operators $H_n:=H\upharpoonright_{\Set{N=n}},\,n=0,1,\ldots,$ are symmetric and bounded, they are  self-adjoint. Hence} so is $H=\oplus_{n=0}^\infty H_n$ as an infinite direct sum of self-adjoint operators. Therefore the propagator $e^{-itH}$ is well-defined for every $t\in\Rb$. 

It is convenient to extend the state space $\cF$ by going to the space $S(\cF)$ of density operators on $\cF$, i.e.~positive trace-class operators $\rho$ on $\cF$, which we identify with  positive linear  functionals (i.e.~expectations)  
of observables,
 $\om(A)\equiv \om_\rho(A):=\Tr(A \rho)$. Consequently, 
  we pass from the Schr\"odinger equation $i\p_t\psi=H\psi$ 
   on $\cF$, to the von Neumann equation  
\begin{equation}\label{vNL}\p_t\rho_t =-i[H, \rho_t] \quad \text{ or } \quad \p_t\om_t(A)=\om_t(i [H, A]).\end{equation}   
The domain of $\Ad_H:A\mapsto [A,H]$ in the space $S(\cF)$ of density operators over the Fock space $\cF$ is given by
		\begin{align}\label{domain_def}
	\cD&:=\{\rho\in S(\cF)\mid \rho\cD(H)\subset\cD(H)\text{ and }[H,\rho]\in S(\cF)\},
\end{align}
{We write $\om\in\cD$ if $\om=\om_\rho$ for some $\rho\in\cD$. }
For each $\rho\in\cD$, the Cauchy problem \eqref{vNL} with initial configuration $\rho$ has a unique solution given by
\begin{equation}\label{1.5'}
	\rho_t\equiv \al'_t(\rho):=e^{-itH}\rho e^{itH}.
\end{equation}
It is straightforward to check that the evolution \eqref{1.5'} preserves total probability and positivity, i.e.,
\begin{equation}
	\label{1.6}
	\Tr(\rho_t)\equiv \Tr(\rho)\quad\text{and}\quad \rho\ge0\implies \rho_t\ge0\quad (t\in \Rb)	,
\end{equation}
as well as the eigenvalues of $\rho$. 

The evolution of observables, dual to $\al_t'$ in \eqref{1.5'} w.r.t.~the coupling $(A,\rho)\mapsto\Tr(A\rho)$, is given by
\begin{equation}\label{1.5}
	A_t\equiv \al_t(A):=e^{itH}Ae^{-itH}.
\end{equation}
In terms of linear functionals with initial condition $\om$, evolution \eqref{1.5'} becomes  $\om_t (A)=\om(A_t)$, where  $\om_t=\om \circ \al_t$, and relations \eqref{1.6} become $\om_t(\one)\equiv \om(\one)$ and $\om\ge0\implies \om_t\ge0$.


Below, 
 we evaluate our inequalities on states $\om$ (which we also consider as initial conditions for \eqref{vNL}) 
satisfying the following conditions:
\begin{equation}\label{g0-cond}
\om \in\cD,\quad	\om (N^2)<\infty,
\end{equation}
where, recall, $N\equiv N_\Lam$ is the total number operator. In what follows, for a subset $X\subset \Lam$, we denote by 
$X^\cp:=\Lam\setminus X$ its complement in $\Lam$, $d_X(x)\equiv \dist(\Set{x},X):=\inf_{y\in X}\abs{x-y}$ the distance function to $X$, $X_\xi:=\Set{x\in\Lam:d_X(x)\le \xi}$ (see \figref{fig:Xxi} below), 
  and $X_\xi^\cp$ is always understood as $(X_\xi)^\cp$. 
\begin{figure}[H]
	\centering
	\begin{tikzpicture}[scale=.8]
		\draw  plot[scale=.52,smooth, tension=.7] coordinates {(-3,0.5) (-2.5,2.5) (-.5,3.5) (1.5,3) (3,3.5) (4,2.5) (4,0.5) (2.5,-2) (0,-1.5) (-2.5,-2) (-3,0.5)};
		
		\draw  plot[shift={(-0.2,-0.25)}, scale=.76,smooth, tension=.7] coordinates {(-3,0.5) (-2.5,2.5) (-.5,3.5) (1.5,3) (3,3.5) (4,2.5) (4,0.5) (2.5,-2) (0,-1.5) (-2.5,-2) (-3,0.5)};
		
		\draw [->] (6,.5)--(1.3,.5);
		\node [right] at (6,.5) {$X$};
		
		\draw [->] (6,0)--(2.3,0);
		\node [right] at (6,0) {$X_\xi$};
		
		\node [below] at (2.32,0) {$\underbrace{}_{\xi}$};
	\end{tikzpicture}
	\caption{Schematic diagram illustrating $X_\xi$.}
	\label{fig:Xxi}
\end{figure}
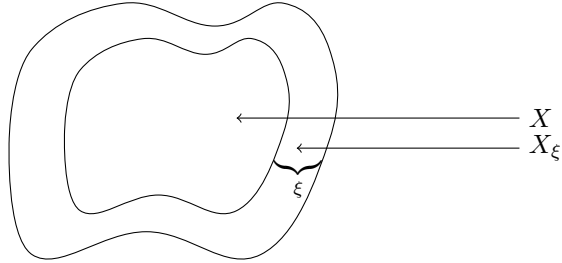
Finally, let
\begin{equation}\label{kappa}
	\kappa\equiv \kappa_0=\sup _{x\in\Lam} \sum_{y\in\Lam} \abs{h_{xy}}\abs{x-y}.
\end{equation}
{The number $\kappa$ bounds the norm of the $1$-particle group velocity operator $i[h,x]$,  see \remref{remk} below.}

\subsection{Maximal velocity bound}\label{sec:MVB}

The main result in this section gives an estimate on the {maximal velocity of propagation of particles into empty regions}. {Continuing with}  terminology of \cite{SigSof2, FLS,FLS2}, 
we call such an estimate the \textit{maximal velocity bound} (MVB).

\begin{theorem}[MVB for lattice quantum many-body system]\label{thm1}
	Suppose \eqref{k-cond} holds with some $n\ge1$. Then, for every $c > \kappa$,
there exists $C=C(n,\kappa_n,c )>0$ s.th.~for all $\eta\ge1,\,X\subset \Lam$,
we have the following estimate for all $\abs{t}< \eta/c$:
	\begin{equation}
		\label{MVE} \al_t(N_{X_\eta^\cp})
		\le 
{C	(N_{X^\cp} +\eta^{-n} N)}.
	\end{equation}
\end{theorem}

Theorem \ref{thm1} is  proved in \secref{sec:pfthm1}, with an outline given in \secref{sec:1.6}. Here and below, an operator inequality $A\le B$ means that $\om(A)\le \om(B)$ for all states $\om$ satisfying \eqref{g0-cond}.

Estimate \eqref{MVE} shows that, if the initial condition $\om$ satisfies \eqref{g0-cond} and is localized in $X$\DETAILS{ in the sense that $\om(N_{X^\cp})=0$}, then, up to polynomially vanishing probability tails, the particles propagate within the light cone $$X_{ct}\equiv \Set{d_X(x)\le ct}$$ for every fixed $c>\kappa$ and all $t$. More precisely,  if we assume the initial state satisfies
		\begin{equation}\label{emptyshellcond}
 \om (N_{X^\cp})=0,
		\end{equation}
	and use the observation, due to M.~Lemm, that under \eqref{emptyshellcond},
	\begin{equation}\label{NNX}
		\om(N^p)=\om(N_X^p)\qquad(p=1,2),
	\end{equation}
 we find, for all $\abs{t}<\eta/c$,
\begin{equation}\label{part-numb-est-om} 
\om_t(N_{X^\cp_\eta})=\om (\al_t (N_{X^\cp_\eta}))\le C\eta^{-n}\om (N_X).\end{equation}
Put differently, the probability that particles are transported from $X$ to any test (or probe) domain $Y$ outside the light cone $X_{ct}$  is of the order $O (\eta^{-n})$, where $\eta=\dist(X,Y)$. 



\subsection{Light-cone approximation of evolution \eqref{1.5}}\label{sec:2.2}
The main result of this section and the next one concerns {the evolution of general \textit{local observables.}}
We say that an operator $A$ acting on $\cF$ is \textit{localized} in $X\subset \Lam$ if \begin{equation}\label{A-loc}
	\sbr{A,a_x^\sharp}=0 \quad \forall  x\in X^\cp,
\end{equation} where $a_x^\#$ stands for either $a_x$ or $a_x^*$.
Denote by $\supp A$ the intersection of all  $X$ s.th.~\eqref{A-loc} holds. 
Then $A$ is localized in $X$ if and only if  $$\supp A\subset X.$$ The support of an initially localized observable generally spreads over the entire space immediately for any $t>0$. Nonetheless, in \thmref{thm2} below, we show that the evolution of local observables under \eqref{1.5} can be approximated by {a family of observables localized within the light cone of the initial support. }

To state our result, we introduce some notations. For a subset $X\subset \Lam$, define the \textit{localized evolution} of observables as
\begin{equation}\label{alloc}
	\al_t^X(A):=e^{itH_X}Ae^{-itH_X},
\end{equation}
where $H_X$ is defined by \eqref{1.1} but with $X$ in place of $\Lam$, and the set of operators
\begin{equation}\label{BX}
	\cB_X :=\Set{A\in\cB(\cF):[A,N]=0,\, \supp A\subset X},
\end{equation}
where  $\cB(\cF)$ is the space of bounded operators on $\cF$. One can check using  definitions \eqref{A-loc}, \eqref{alloc}, and the relation $[H_X,N]=0$ that, for all $A\in\cB_X$, the evolution   $\al_t^X(A)$ lies in $\cB_X$ for all $t$. 

The main result of this section is that  the full evolution $
\al_t(A)$ can be well approximated by the localized evolution $\al_t^{X_\xi}(A)$, supported  inside the light cone of $\supp A$:
\begin{theorem}[Light-cone approximation of quantum evolution]\label{thm2}
		Suppose \eqref{k-cond}--\eqref{v-cond} hold with some $n\ge1$. Suppose a state $\om$ satisfies \eqref{g0-cond} and   \eqref{emptyshellcond},
with some $X\subset \Lam$.
{Then, for every $c>2\kappa$, there exists $C=C(n,\kappa_n,\nu_n,c)>0$ s.th.~for all }$\xi\ge1$ and operator $A\in \cB_X$ (see \eqref{BX}),   the full evolution $
\al_t(A)$ is approximated by 
the local evolution $
\al_t^{X_\xi}(A)$ for all $\abs{t}<\xi/c$, as
	\begin{equation}\label{lcae} 
\begin{aligned}
\abs{\om {	\del{\al_t (A)-\al_t^{X_\xi}(A)} }} \le C\abs{t}\xi^{-n}\Norm{A}\om(N_X^2).
\end{aligned}
	\end{equation}

\end{theorem} 
\thmref{thm2} follows from \thmref{thm2-gen}.
We  sketch the proof of \thmref{thm2}  in \secref{sec:4.1}.


 \subsection{Lieb-Robinson-type bounds}\label{sec:LRB}

Using \thmref{thm2}, we prove a Lieb-Robinson-type bound 
for general interacting quantum many-body systems:

\begin{theorem}[Weak Lieb-Robinson bound] 
	\label{thm3}
Suppose the assumptions of \thmref{thm2} hold with $n\ge1,\,X\subset\Lam$. {Then, for every $c>2\kappa$, there exists $C=C(n,\kappa_n,\nu_n,c)>0$ s.th.~for all }$\xi\ge1$,  $Y\subset\Lam$ with $\dist(X, Y)\ge2\xi$, and  operators $A\in \cB_X,\,B\in \cB_Y$, 
we have the following estimate for all $\abs{t}< \xi/c$:
\begin{align}
		\abs{\om \del{[ \al_t (A), B] }}  
	\le C\abs{t}\Norm{A}\Norm{B}\xi^{-n}\om(N^2_X).\label{LRB}
	\end{align}
\end{theorem}

\thmref{thm3} is proved in \secref{sec:6}. We call a bound of the form \eqref{LRB} the \textit{weak Lieb-Robinson bound (LRB)}.  Unlike the classical LRB, estimate \eqref{LRB} depends on a subclass of  states 
and provides power-law, rather than exponential, decay.

	Estimate \eqref{LRB} shows that, with the probability approaching $1$ as $t\ra \infty$, an evolving family of observable $A_t=\al_t(A)$ remains commuting with any other observable supported outside the light cone \[\{x\in \Lambda \,|\, \dist(x, \supp A)\le c t\},\] for any fixed {$c>2\kappa$,} provided the supports of these observables are separated  {by initially empty regions}. \DETAILS{This implies that the maximal speed of quantum propagation for \eqref{1.5} is bounded  (up to an absolute constant) by the number $\kappa$ defined in \eqref{kappa}.}

\subsection{Propagation/creation of correlations}\label{sec:cor}

In this section we address the following questions (c.f.~\cite{BHV,NachOgS}): 
\begin{itemize}
	\item {Assuming the initial state $\om$ is weakly correlated in a domain $Z^\cp\subset \Lam$, 
	how long does it take for the correlations in $Z$ to spread, under the evolution \eqref{1.5}, into $Z^\cp$? Put differently, how long does it take to create correlations in $Z^\cp$?}

\end{itemize}

To begin with, we define what we mean by weakly correlated states.

\begin{definition}\label{def:WC'}

	Let  $Z\subset \Lam$. For subsets $X,\,Y\subset \Lam$,  let $d_{XY}:=\dist(X,Y)$ and  
	\begin{equation}\label{dXYZ}
{		d_{XY}^Z:=\min (d_{XY}, d_{XZ}, d_{YZ}).}
	\end{equation}
We say a state $\om$ is {\it weakly correlated  in a subset $Z^\cp$ at a scale} $\l>0$, or $\WC(Z^\cp, \l)$, if there exists $C>0$ s.th.~for all subsets $X,\,Y\subset Z^\cp$ with {$d_{XY}^Z>0$} and operators $A\in\cB_X,\,B\in\cB_Y$ (see \eqref{BX}), the following estimate holds:
\begin{equation}\label{221'}
	|\om^c(AB)|\le C\norm{A}\norm{B} (d_{XY}^Z/\l)^{1-n}, 
\end{equation}
	where $\om^c(A,B):=\om(AB)-\om(A)\om(B)$ is a ($2$-point) {\it connected correlation function}. 
	

\end{definition}


  Clearly,  $\l$ plays the same role as the correlation length for exponentially decaying correlations. The main result of this section, proved in \secref{sec:pfthm8.2}, shows that the maximal speed for the propagation/creation of correlations  is bounded by $3\kappa$:
\begin{theorem}[Propagation/creation of correlation]\label{thm5}
		Suppose \eqref{k-cond}--\eqref{v-cond} hold with some $n\ge1$. 	Let  $Z \subset \Lam$ and suppose the initial state $\om$ is $\WC(Z^\cp, \l)$  and satisfies \eqref{g0-cond} and
	\begin{equation}\label{emptyshellcond'}
		\om(N_{Z^\cp})=0.
	\end{equation}	 Then,  $\om_t$ is  $\WC(Z^\cp, 3\l)$ 
	for 
	all $ \abs{t}<\l/3\kappa$; specifically, for any operators $A\in\cB_X,\,B\in\cB_Y$	supported in $X,\,Y\subset Z^\cp$ with $d_{XY}^Z>0$ and for $\abs{t}<\l/3\kappa$,
\begin{equation}\label{221}
	|\om^c_t(AB)|\le C\norm{A}\norm{B} (d_{XY}^Z/3\l)^{1-n}\om(N_Z^2) ,
\end{equation}	  For short-range (i.e.~exponentially decaying) interactions, \eqref{221} holds 
for all $n\ge1$. 
\end{theorem}

For the second statement, we note that for short-range interactions, conditions \eqref{k-cond}--\eqref{v-cond} are valid for all $n$.

\subsection{Constraint on the propagation of quantum signals}\label{sec:6.1}

The weak LRB  \eqref{LRB} imposes a direct constraint on the propagation of information through the quantum channel defined by the time evolution $\al_{t}'$ of quantum states (see e.g.~\cite{BHV, FLS2, Pou}). For example,  assume that Bob at a location $Y$ is in possession of a state $\rho$ and an observable $B$ and would like to send a signal through the quantum channel $\alpha_{t}'$ to Alice who is at $X$ and who possesses the same state $\rho$ and an observable $A$. {To send a message}, Bob uses $B$ as a Hamiltonian to evolve $\rho$ for a time   $r>0$, and then sends Alice the resulting state $\rho_r=\tau_{r}(\rho)$, where $\tau_{r}(\rho):=e^{-iBr}\rho e^{iBr}$, 
as $\alpha_{t}'(\rho_r)$. 
To see whether Bob sent his message, 
Alice computes the difference between the expectations of $A$ in the states $\al_{t}'(\rho_r)$ and $\al_{t}'(\rho)$, which we call the signal detector, SD: 
\begin{align} \label{CommEstPhys'} \SD(t,r):=\tr\big[A\al_{t}'(\rho_r)-A\al_{t}'(\rho) \big].&
\end{align} 
The main result of this section gives an upper bound on this difference:

\begin{theorem}\label{thm6}

Let the assumptions of \thmref{thm2} hold with $n\ge1,\,{X\subset\Lam}$ and $\om(\cdot)=\Tr((\cdot)\rho)$.
{{Then, for every $c>4\kappa$, there exists $C=C(n,\kappa_n,\nu_n,c)>0$ s.th.~for all } $\xi\ge2$,  $X,\,Y\subset\Lam$  with {$\dist(X, Y)\ge2\xi$,} and  operators $A\in\cB_X$, $B\in \cB_Y$, with the operator kernel of $B$ satisfying  \eqref{k-cond} with $n\ge1$,}
 we have the following estimate for all $r,\abs{t}< \xi/c$:
	\begin{equation}\label{7.2}
		\abs{	\SD(t,r)}\le Cr\abs{t}\xi^{-n}\Norm{A}\Norm{B}\Tr(N_X^2\rho).	\end{equation}
\end{theorem}

The proof of this theorem is found in \secref{sec:pfthm6}.

\subsection{Bound on quantum state control} 
\label{sec:6.2}
In this section, we  derive a bound on the information-theoretic task of 
state control. {For any subset $S\subset \Lam$, we denote by  $\cF_{S}$  the  Fock space over the one-particle Hilbert space $\ell^2(S)$, see Appendix \ref{sec:Fock-sp} for the definitions and discussions. Due to the tensorial structure $\cF \simeq \cF_Y\otimes \cF_{Y^\cp}$ (see \eqref{F-factor}), we can define  the partial  trace $\mathrm{Tr}_{\cF_{Y^c}}$ over ${\cF_{Y^c}}$, e.g. by the equation $\mathrm{Tr}_{\cF_{Y}}(A\mathrm{Tr}_{\cF_{Y^c}}\rho)=\mathrm{Tr}((A\otimes \one_{\cF_{Y^c}})\rho)$ for every bounded operator $A$ acting on $\cF_{Y}$. This allows one to define  a {\it restriction} of a state $\rho$ to the density operators on the local Fock  space $\cF_{Y}$, $Y\subset \Lam,$ 
	by $\rho_Y:= \mathrm{Tr}_{\cF_{Y^c}} \rho$. 
	}
	
	Let  $\tau$ be a quantum map (or {\it state control
		map}) 
	supported in $X$. 
	Given  a density operator $\rho$, our task is to design $\tau$ so that  at some time $t$, the evolution $\rho^\tau_t:=\al_t(\rho^\tau)$  of  the density operator $\rho^\tau:= \tau(\rho)$  has the restriction $[\rho^\tau_t]_{Y}$ to $S(\cF_{Y})$, which is close to a desired state, say $\s$.

	To measure the success of  the  transfer operation, one can use the figure of merit  
	\begin{equation}\label{fom1}
		F([\rho^\tau_t]_Y, \s),
	\end{equation}
	where $F(\rho,\sigma)=\|\sqrt{\rho}\sqrt{\sigma}\|_{S_1}$ is the {\it fidelity}. 
	Here $\|\rho\|_{S_1}$ denotes the Schatten $1$-norm.
	Note that $0\le F(\rho,\sigma)\le 1$, with $F(\rho,\sigma)=1$ if and only if $\rho=\sigma$. If $\s=\abs{\phi\rangle\langle\phi}$, then $F(\rho,\si)=\sqrt{\br{\phi,\rho\phi}}$ and therefore
	$F(\rho,\sigma)=0$ if  $\rho\phi=0$.
	So, one would like to find $\tau$ maximizing \eqref{fom1}. 
	Using this figure of merit, one might be able to estimate the upper bound on the state transfer time.
	
	On the other hand, to show that the state  transfer is impossible in a given time interval, we would compare $\rho^\tau_t$ and $\rho_t:=\alpha_t(\rho)$ by using  (c.f.~\cite{EpWh,FLS2}) \begin{equation}\label{fom2}
		F([\rho^\tau_t]_Y, [\rho_t]_Y),
	\end{equation}
	as a figure of merit, and try to show that it is close to $1$ for $t\le t_*$ and for all state preparation (unitary) maps  $\tau$   localized in $X$. If this is true, then clearly using $\tau$'s   localized in $X$ does not affect states in $Y$.

\DETAILS{Following \cite{EpWh}, we use the figure of merit to quantify the success of the state control
	\begin{equation}\label{eq:fom}
		F(\mathrm{Tr}_{Y^\cp} \alpha_t(\rho),\mathrm{Tr}_{Y^\cp} \alpha_t(A\rho A^*))
	\end{equation}
	where $F(\xi,\si)$ is given, for non-negative, trace class operators $\xi,\si,$, by
	$$F(\xi,\sigma):=\norm{\sqrt{\si}\sqrt{\xi}}_{S^1}=\Tr\del{\sqrt{\sqrt\xi \si\sqrt \xi})}$$ is the fidelity, where $\norm\cdot_{S^1}$ denotes the Schatten $1$-norm. 
	For  two normalized density operators $\xi,\si$,  the fidelity $F(\xi,\si)=1$ if and only $\xi=\sigma$. 
	{Clearly,  estimate \eqref{7.2} being close to $1$ independent of $\tau$ means that changing the state $g$ in $X$ does not affect it in $Y$ at time $t$.}}

Specifically, we take $\tau$ to be of the form $\tau(\rho)=U\rho U^*\equiv \rho^U$, where $U$ is a unitary operator. 
 Our result in this setting is the following lower bound on  the fidelity of quantum state control:

\begin{theorem}[Quantum control bound]\label{thm:qst}

Suppose the assumptions of \thmref{thm2} hold with $n\ge1,\,X\subset\Lam$, and $\om(\cdot)=\Tr((\cdot)\rho)$, where $\rho$ is a pure state.
{Then, for every $c>8\kappa$, there exists $C=C(n,\kappa_n,\nu_n,c)>0$ s.th.~for all } $\xi\ge4$, $Y\subset \Lam$ with $\dist(X,Y)\ge2\xi$, and unitary operator	 $U\in\cB_X$ (see \eqref{BX}), we have the following lower bound for all $\abs{t}< \xi/c$:
	\begin{equation}\label{eq:qst}
		\begin{aligned}
			F(\mathrm{Tr}_{Y^\cp} (\al'_{t}(\rho)),\mathrm{Tr}_{Y^\cp}( \al'_{t}(\rho^U))) \geq& 1- C\abs{t}\xi^{-n}\Tr(N_X^2\rho).
		\end{aligned}
	\end{equation}
\end{theorem}

The proof of this theorem is found in \secref{sec:pfthm:qst}
	As noted at the begining of this section, \thmref{thm:qst} imposes a constraint on the best-possible quantum state transfer protocols for the quantum many-body dynamics.

\subsection{Spectral gap and decay of correlation}\label{sec:2.8}


Denote by $\Om$ the normalized ground state of the Hamiltonian $H$ in \eqref{1.1}\DETAILS{ (i.e.~$H\Om=0,\,\norm{\Om}_\cF=1$)}. 
The main result of this section is the following:
\begin{theorem}[Gap at the ground state implies decay of ground state correlations]
	\label{thm:gap}
	Suppose $H$ in \eqref{1.1} has a spectral gap of size $\g>0$ at the ground state energy $E$. 
Suppose the assumptions of \thmref{thm2} hold with $n\ge1,\,X\subset\Lam$,
{and $\om=\Ondel{(\cdot)}.$} {Then, there exists $C=C(n,\kappa_n,\nu_n)>0$ s.th.~for all } $\xi\ge1$,  $Y\subset\Lam$ with $\dist(X, Y)\ge2\xi$, and  operators $A\in \cB_X,\,B\in \cB_Y$, we have the following bound:
	\begin{equation}\label{281}
		{\abs{\Ondel{BA}}}\le C\norm{A}\norm{B}(\g^{-1}\xi^{-2}+\,{\xi^{1-n}\Ondel {N_X^2}}).
	\end{equation}
\end{theorem}

This theorem is proved in \secref{sec:pfthm:gap}. 

\subsection{Light cone in macroscopic particle transport}\label{sec:MPT}
In this section, we derive an estimate on the macroscopic particle transport for states evolving according to \eqref{vNL}. To begin with, for a given subset $S\subset \Lam$, we define  the (macroscopic)  local relative particle numbers as
\begin{equation}\label{6.1}
	\bar N_S:=\frac{N_S}{N_\Lam}. 
\end{equation}
For $0\le \nu\le 1$, we write $P_{\bar N_S\ge \nu}$ 
for the spectral projection associated to the self-adjoint operator $\bar N_S$ onto the spectral interval $[\nu,1]$. 

The main result of this section is the following:
\begin{theorem}\label{thm4}
	Suppose \eqref{k-cond} holds with some $n\ge1$. Suppose the initial state $\om\in\cD$  satisfies  
	\begin{equation}\label{6.0}
		\om(P_{\bar N_{X^\cp}\ge \nu})=0,
	\end{equation}
	with some $\nu\ge0,\, X\subset \Lam$.{ Then, for all $	\nu'>\nu$, $c>\kappa,$ there exists $C=C(n,\kappa_n,c,\nu'-\nu)>0$ s.th.~}for every $\eta\ge1$, we have the following estimate for all $\abs{t}< \eta/c$:
	\begin{equation}\label{6.2}
		\om_t\del{P_{\bar N_{X_\eta^\cp}\ge \nu'}}\le C\eta^{-n}.
	\end{equation}
\end{theorem}

\thmref{thm4} is proved in \secref{sec:pfthm4}.
Note that estimate \eqref{6.2} holds for rather general initial states (including ones with particle densities uniformly bounded from below) and 
it	controls macroscopic fractions of particles. 

\subsection{Discussions of the results and extensions}\label{sec:ex}

We begin with a number of remarks on Theorems \ref{thm1}--\ref{thm3}.


\begin{remark}
	Recall that the grid size for the underlying lattice is fixed throughout the paper. Consequently, estimates obtained in this paper are all implicitly dependent on the grid size and in general blow up as the latter shrinks to $0$. This has to do with the implicit momentum cutoff baked into the discretization process. 
\end{remark}

\begin{remark}
	At the quantum energies in nature and laboratories (besides particle accelerators), the maximal speed of propagation implied by \thmref{thm1} is much below the speed of light, so the non-relativistic nature of Quantum Mechanics is unimportant here.
\end{remark} 

\begin{remark}
	The factor $\om(N_X^2)$ in Thms.~\ref{thm3}-\ref{thm:gap} originates in \thmref{thm2}.
\end{remark}

\begin{remark}
	The conclusion of \thmref{thm4} is thermodynamically stable, in the sense that it does not change as $\abs{\Lam}$, $\om(N)\to \infty$ with $\om(N^p)\le C\abs{\Lam}^p$, $p=1,2$. 
\end{remark}

\begin{remark}
		From the proof of \thmref{thm1}, one can see that the constant in \eqref{MVE} is inversely proportional to the difference $c-\kappa>0$. See the proof of \propref{prop:propag-est1}, particularly estimate \eqref{propag-est3}. 
\end{remark}

\begin{remark}\label{remk}
	The conclusion of \thmref{thm1} holds under a slightly weaker assumption. Indeed, let $h$ be the $1$-particle Hamiltonian in \eqref{1.1}. Then condition \eqref{k-cond} implies that there exists $C=C(n)>0$ s.th.~for every subset $X\subset \Lam$, the multiple commutators $\ad{p}{d_X}{h}$ satisfies
	\begin{equation}\label{h-cond}
		\kappa_p':=	\norm{\ad{p+1}{d_X}{h}}\le C\quad (p=0,1,\ldots,n).
	\end{equation}
	This important consequence is proved in \lemref{lemA.1}. The statement of \thmref{thm1} is valid under assumption \eqref{h-cond}, with $\kappa\equiv \kappa_0$ replace by
	\begin{equation}\label{kappa'}
		\kappa'\equiv \kappa_0'=\norm{i[h,d_X]}.
	\end{equation}
{Here $i[h,d_X]$ is related to the  the group velocity operator $i[h, x]$, where $x$ is the $1$-particle position observable.}
\end{remark}

\begin{remark}\label{rem4}
	A sufficient condition for \eqref{k-cond} (and therefore the weaker condition \eqref{h-cond}) is $$\abs{h_{xy}}\le C(1+\abs{x-y})^{-(n+d+2)}$$
	and similarly for \eqref{v-cond}.
	

\end{remark}

\begin{remark}\label{rem:lit}
\cite{SHOE} obtains a result similar to   \thmref{thm1}, but with the exponential error bound, while having an additional prefactor (coming from the summation over the sites of $X_\eta^\cp$), essentially, $|\Lam|$. Apart from Theorems \ref{thm4}, whose proof uses results of the proof of Theorem \ref{thm1}, we could have used this result in Theorems \ref{thm2}, \ref{thm3}, \ref{thm5}-\ref{thm:gap}, instead of Theorem \ref{thm1} to obtain the exponential decay with the prefactor $|\Lam|$, instead of the power-law decay. We also note that the result of \cite{SHOE}  requires $h=-\Delta$ (the negative of lattice Laplacian) and uses a bound on the matrix of the imaginary time propagator $e^{\tau\Delta t}$, while our analysis requires only commutators of $h$ with (functions of) $x$.

 In \cite{KS2}, the authors derive an approximation result similar to that of \thmref{thm2} but with a logarithmically modified light cone and with \eqref{emptyshellcond} replaced by the low-density condition $\sup_{x\in\Lam}\om (e^{c_0n_x})\le M$. 
 
  In \cite{YL}, 
  a weak LRB similar to \thmref{thm3} is proven for 
   the steady state $ e^{-\mu N}$.

\DETAILS{
Moreover, for  the Bose-Hubbard model,  MVBs similar to \thmref{thm1} were obtained in \cite{SHOE}, with exponentially decaying leaking probability tails. This result  establishes linear light cone for the Bose-Hubbard model with initially localized states.
 In \cite{KS2},  the authors derived an almost-linear light-cone approximation  similar to our \thmref{thm2}.  In \cite{YL}, the authors obtained weak LRBs similar to \thmref{thm3} for commutators tested against the steady state $ e^{-\mu N}$. See \remref{rem:lit} for more mathematical details of these results. 
In principle, the bounds in these works could be extended to more general models apart from the Bose-Hubbard one, but it seems that finite-range assumption on the interactions was essential in these works. }

\DETAILS{In view of \thmref{thm3},	for localized observables $A,\,B$ separated by a distance $\eta>0$, the authors of \cite{YL} proved that $\om_\mathrm{st}([A_t,B])\le c(vt/\eta)^{c'\eta}\om_\mathrm{st}(\one)$, where $c,\,c'$, and $v$ are constants and $c,\,v$ depend on the observables $A,\,B$.} 
	
\end{remark}

\begin{remark}
	Estimate \eqref{part-numb-est-om} and other results can be extended in a straightforward way to initial states of the form $\om=\al\om_*+\beta\om'$, where $\om_*$ is a stationary state, $\om'$ satisfies \eqref{emptyshellcond}, and $\al,\,\beta\ge0$ with $\al+\beta=1$. Then \eqref{part-numb-est-om} implies that $0\le \om_t(N_{X^\cp_\eta})-\al\om_*(N_{X^\cp_\eta})\le C\eta^{-n}\om'(N_X).$
\end{remark}

\begin{remark}
\thmref{thm2} follows from the following result, proved in Appendix \ref{sec:4.11}: 
\end{remark}

\begin{theorem}\label{thm2-gen}
	Suppose \eqref{k-cond}--\eqref{v-cond} hold with some $n\ge1$. 
 Let a  state $\om $ satisfy \eqref{g0-cond}. 
Then, for every $c>2\kappa$, there exists $C=C(n,\kappa_n,\nu_n,c)>0$ s.th.~for all $\xi\ge1$,  $X,\,Y\subset\Lam$ with $\dist(X, Y)\ge2\xi$, and  operators $A\in \cB_X,\,B\in \cB_Y$, 
	we have, for all $\abs{t}< \xi/c$:
\begin{align}\label{lcae-gen}
			|\om\big(B&(\al_t(A)-\al_t^{X_\xi}(A))\big)|\notag\\
&\le C\abs{t} \Norm{A}\Norm{B} \del{ \om\del{N_{X_{2\xi}\setminus X}N}+\xi^{-n}\om(N^2)}.\end{align}
	\end{theorem}
Note that $\om(N_SN)\ge0$ since $N$ and $N_S$ commute. 

\begin{proof}[\thmref{thm2-gen} implies  \thmref{thm2}:] 
 	Applying estimate \eqref{lcae-gen} 
 	with $B=\one$, and the relations $N_{X_{2\xi}\setminus X}\le N_{X^\cp}$ and \begin{align}\label{2.32s}
 		\om( N_{X^\cp}N)\le\om(N^2)^{1/2}\om(N_{X^\cp}^2)^{1/2},
 	\end{align}and using condition \eqref{emptyshellcond}, we find the desired estimate \eqref{lcae}.
\end{proof}

We compare \eqref{lcae-gen} with $B=\one$ with the corresponding result for the Hubbard model, i.e. $v_{xy}=\lam \delta_{xy}$ for some $\lam\in \R$, see \eqref{lcae-gen-Hub'} below.

Next, we comment on various extensions of the results from preceding sections. 
Theorems \ref{thm1}-\ref{thm:gap} can be extended to (a) time-dependent one-particle and two-particle operators $h$ and $v$ satisfying \eqref{k-cond}--\eqref{v-cond} uniformly in time; (b) $k$-body potentials (added to or replacing the second term on the r.h.s. of \eqref{1.1})
\[V= \sum_k\sum_{x_1 \dots x_k} \prod_i a^*_{x_i}v_{x_1 \dots x_k}\prod_i a_{x_i}\] under appropriate conditions on $v_{x_1 \dots x_k}$.



Through \corref{cor3} below, 
Theorems \ref{thm3}--\ref{thm:gap} can be generalized to relative $N^{\nu/2}$-bounded observables with $0<\nu<\infty$. By definition, this class of operators contains all polynomials in $\Set{a_x,\,a_x^*}_{x\in\Lam}$ with degree at most $\nu$. Precise definitions and further comments are delegated to Appendix \ref{sec:nu}.

\begin{corollary}
	\label{cor3}
	Suppose \eqref{k-cond} holds with some $n\ge1$. Let $\nu,\,q\ge0.$ Then, for all $c > \kappa$,
	there exists $C=C(n,\kappa_n,c)>0$ s.th.~for every $\eta\ge1 $ and two subsets $X\subset S\subset \Lam$,
	we have the following estimate for all $\abs{t}< \eta/c$:
	\begin{align}\label{MVE-nuq}
		\al_t^S(N_{S\setminus X_\eta}^{q+1}N^{\nu})\le &C
		\del{N_{S\setminus X}N^{\nu+q} +\eta^{-n} N^{\nu+q+1}},
	\end{align}
where $\al_t^S$ is as in \eqref{alloc}.
\end{corollary}
\begin{proof}
	
	We use Theorem \ref{thm1} with $\al^S_t(\cdot)$ in place of $\al_t(\cdot)$, which is possible because $H_S$ also satisfies \eqref{k-cond} with the same $n\ge1$ as in the assumption. This gives estimate \eqref{MVE} with $S\setminus (\cdot)$ in place of $(\cdot)^\cp$. This, together with the relations $N_{S\setminus X_\eta}^{q+1}\le N_{S\setminus X_\eta}N^{q}$, $[N,H_S]=0$, and $N_S\le N$, implies  that
	$$
	\begin{aligned} 
		\al_t^S\del{N_{S\setminus X_\eta}^{q+1}N^{\nu}}\le &\al_t^S\del{N_{S\setminus X_\eta}}N^{\nu+q}
		\\\le&	C \del{N_{S\setminus X} +\eta^{-n} N_S}N^{\nu+q} \\
		\le&C\del{N_{S\setminus X}N^{\nu+q} +\eta^{-n} N^{\nu+q+1}}.
	\end{aligned}
	$$
	This gives \eqref{MVE-nuq}. 
\end{proof}

\section{Proof of \thmref{thm1}}\label{sec:pfthm1}

\subsection{Outline of the proof of \thmref{thm1}}\label{sec:1.6}
The proofs of Theorems \ref{thm2}--\ref{thm3} and the subsequent applications are based on \thmref{thm1}, whose proof we outline now. 

\subsubsection{Propagation identifier observables.}
 Recall that the second quantization $\dG$ of $1$-particle operators on $\hf \equiv \ell^2(\Lambda)$ 
is given by  
\begin{equation}\label{2nd-quant}\dG(b):=\sum_{\Lambda\times\Lam} b_{xy} a_x^* a_y , \end{equation}
where $b_{xy}$ is the matrix (``integral'' kernel) of an operator $b$ on $\ell^2(\Lam)$. 
As we identify a function $f:\Lambda\to\Cb$ with the multiplication operator induced by it on $\hf\equiv \ell^2(\Lam)$, we write
\begin{equation}\label{1.7}
	\hat f\equiv  \dG(f):=\sum_{x\in\Lambda}f(x)a_x^*a_x. 
\end{equation}
We denote by $\chi_S^\sharp$ the characteristic function of a subset $S\subset\Lambda$. For  $f=\chi_S^\sharp$, the above gives the local particle number operators  $N_S\equiv \dG(\chi_S^\sharp) $ in \eqref{1.4}. 
For a differentiable real function $f$, we write $\hat f'\equiv \dG(f')$ and $\hat f_{ts}'\equiv  \dG(f_{ts}')$, where $f_{ts}'\equiv f'\del{\tfrac{d_X-vt}{s}}$.

As in \cite{FLS, FLS2}, we control the time evolution  associated to \eqref{1.1} by monotonicity formulae for a class of observables called  \textit{adiabatic spacetime localization observables (ASTLOs)}, defined as
\begin{equation}\label{PIO1}
	\hat\chi_{t s}:=\dG(\chi_{t s}). 
\end{equation}
Here $ s>0,\,t\in\Rb$, and $\chi_{ts}$ is the family of multiplication operators by real functions
\begin{equation}\label{chi-ts}
	\chi_{ts}=\chi\del{\frac{d_X-vt}{s}}, 
\end{equation}        
where $d_X$ is the distance function to $X$, $v\in(\kappa,c)$, with   $\kappa$  from \eqref{kappa} and $c$ from the statement of \thmref{thm1}. We assume that $\chi$ belongs to the following set of functions:
\begin{equation}\label{F}
	\begin{aligned}
		\cX\equiv \cX_\delta
		:=&\Set{\chi\in C^\infty(\R)\left|
			\begin{aligned}
				&\supp \chi\subset \Rb_{\ge0},\;\supp \chi'\subset (0,\delta)\\
				&\chi^\prime\ge 0,\,\sqrt{\chi'}\in C^\infty(\R)
			\end{aligned}\right.
		},
	\end{aligned}
\end{equation}
for some  $\delta>0$. 
Later on, we will choose the number $\delta$ in \eqref{F} as $\delta=c-v$ with $c$ and $v$ given in the statement of \thmref{thm1} and \eqref{chi-ts}, respectively.  We note that $\chi\ge0$ for each $\chi\in\cX$.  Additional properties of $\cX$ will be stated in \secref{sec:pfMVE}.
	Physically, 
	$\hat\chi_{t s}$  is a smoothed local particle number operator, measuring fraction of the particles outside the light cone of $X$. We also write $\chi'_{ts}=(\chi')_{ts}$.

\subsubsection{ Recursive monotonicity formula} 

For a differentiable family of observables, define the Heisenberg derivative 
\begin{align}
	\label{Heis-der}&D A(t)=\frac{\partial}{\partial t}A(t)+i[H, A(t)],
\end{align}
so that
\begin{align}	\label{dt-Heis}
	&\di_t\al_t(A(t)) =\al_t(DA(t))\iff	\di_t\om_t(A(t)) =\om_t(DA(t)),
\end{align}
where $\om_t=\om\circ \al_t$ is the evolution of state associated to \eqref{vNL} with initial state $\om$.
We will use the identity \eqref{dt-Heis} to prove a key differential inequality:

\begin{theorem}[Recursive monotonicity of $\astlo_{ts}$]\label{thm:RME} 
	Suppose the assumptions of \thmref{thm1} hold.
	Then, for every $\chi\in\cX$, there exist $C=C(n,\kappa_n,\chi)>0$ and, if $n\ge2$, $\xi^k=\xi^k(\chi)\in\cX,\,k=2,\ldots, n$, each supported in $\supp \chi$, s.th.~{for all $s>0,\,t\in\Rb$,}
	\begin{align}\label{RMB}
		D\hat\chi_{ts} 
		\leq &-\frac{v -\kappa}{s} \astlo'_{ts}+C\ {\sum_{k=2}^{n}s^{-k}\widehat{(\xi^k)'}_{ts}+ Cs^{-(n+1)}N}.
	\end{align}
(The sum in the r.h.s. is dropped if $n=1$.)
\end{theorem}
\thmref{thm:RME} is deduced at the end of this section. 
Since the second  term on the r.h.s.~is of the same form as the leading, negative term {(recall $v>\kappa$ in \eqref{chi-ts})}, estimate \eqref{RMB} can be bootstrapped to obtain an integral inequality with $O(s^{-n})$ remainder, see \propref{prop:propag-est1} below. Hence, we call \eqref{rme} the \textit{recursive monotonicity estimate}.

For the next step, we observe that the second quantization \eqref{1.7} has the properties
\begin{align} \label{dG-lin}&\dG(v+cw)=\dG(v)+c\dG(w),\\
	\label{dG-adj}&\dG(v^*) = \dG(v)^*,\\
	\label{dG-pos}&\dG(v) \le \dG(w) \Longleftrightarrow v\le w,\\
	\label{dG-com1}&\ad{k}{\dG(v)}{\dG(w)}=\dG(\ad{k}{v}{w}),\quad k=1,2,\ldots,
\end{align} 
for all $1$-particle operators $v$ and $w$ acting on $\hf$ and scalars $c$. 
These properties are either obvious or are obtained by direct computation using the canonical commutator relations, see e.g.~\cite[p.9]{FLS}.
Moreover, the second term on the r.h.s. of \eqref{1.1}, which we  denote by $V$,  satisfies 
\begin{align}
	\label{V-cond}[V,\dG(f)]=0\quad \forall f\in \ell^\infty(\Lam). 
\end{align}

Relations \eqref{dG-lin}--\eqref{V-cond} allow us to reduce estimates on $D\hat\chi_{ts}$ to those on $d \chi_{ts}$, where $d b$ is the $1$-particle Heisenberg derivative,  defined as 
\begin{align}
	\label{d-Heis}
	&	d b(t):=\p_t b(t)+i[h, b(t)],
\end{align}
 for  a differentiable path of $1$-particle operator $b(t)$ on $\hf$.
Indeed, let $H_0:=\dG(h)$ and define  the free Heisenberg derivative as
\begin{align}
	\label{freeHeis-der}&D_0 A(t)=\frac{\partial}{\partial t}A(t)+i[H_0, A(t)].
\end{align}
Then, by \eqref{dG-com1}, we have $D_0\dG\big(b\big)=\dG\big(d b\big)$. This, together with property \eqref{V-cond}, gives
\begin{align} \label{D-d-rel}
	D\dG\big(f\big)=\dG\big(d f\big)\iff D\hat f=\widehat{d f},
\end{align}
for every multiplication operator (by a function) $f$. 

In \secref{sec:pfRME}, we prove the following:
\begin{proposition}\label{prop:RME-rme-red} 
	Suppose the assumptions of \thmref{thm1} hold.
Then,  for every $\chi\in\cX$,  there exist $C=C(n,\kappa_n,\chi)>0$ and, if $n\ge2$, $\xi^k=\xi^k(\chi)\in\cX,\,k=2,\ldots, n$, each supported in $\supp \chi$, s.th.~{for all $s>0,\,t\in\Rb$,}
	\begin{align}\label{rme}
		d\chi_{ts} 	\leq &-\frac{v -\kappa}{s} \chi'_{ts}+C\del{\sum_{k=2}^{n}s^{-k}(\xi^k)'_{ts}+ s^{-(n+1)}}.
	\end{align}
(The sum in the r.h.s. is dropped if $n=1$.)
\end{proposition} 
This proposition, together with relation \eqref{D-d-rel}, implies \thmref{thm:RME}.

\subsection{Proof of Theorem \ref{thm1} assuming \propref{prop:RME-rme-red}}
\label{sec:pfMVE}


Recall that $\chi^\sharp_S,\,S\subset \Lam$ denotes the characteristic function of $S$.
The main result of this section is the following:

\begin{theorem}
	\label{thm:msb-cond}

{		Let the assumptions of \thmref{thm1} hold.
Suppose \propref{prop:RME-rme-red} holds and, for all $\chi\in \cX$ with $\norm{\chi}_{L^\infty}=1$, $s=\eta/c$, and $\abs{t}< s$,} the following holds: 
\begin{align}\label{chi-0s-est} 
	&\chi_{0s}  \le  \chi^\#_{X^c},\\
	\label{chi-ts-est}&\chi^\#_{X_\eta^\cp} \le \chi_{ts}.
\end{align}
{Then the conclusion of Theorem \ref{thm1} holds.} \end{theorem} 
The proof of \thmref{thm:msb-cond} is found at the end of this section.
It uses the following properties of the set $\cX$ from \eqref{F}:
\begin{enumerate}[label=(X\arabic*)]
	\item \label{X1}If $w\in  C^\infty$ and  $\supp w\subset (0,\delta)$, then the antiderivative $\int^x w^2\in \cX$. 
	\item \label{X2}If $\xi_1,\ldots,\xi_N\in\cX$, then there exists  $\xi\in \cX$ s.th.~$\xi_1+\ldots +\xi_N\le \xi$. 
\end{enumerate}

{For the $1$-particle Hamiltonian $h$ in \eqref{1.1} and operators $b$ acting on $\hf$, let $\beta_t$ be the $1$-particle evolution
\begin{equation}\label{beta-def}
	\beta_t(b):=e^{ith}be^{-ith},
\end{equation} c.f.~\eqref{1.5},} and note that \begin{align}\label{3.7'}
\di_t\beta_t(b(t))=\beta_t(db(t)),
\end{align} c.f.~\eqref{dt-Heis}.  We denote
\begin{align}\label{2.2}
	\chi_{s}(t):= \beta_t(\chi_{ts})\ \qquad \text{ and }\ \quad  \chi_{s}'(t):= \beta_t(\chi_{ts}').\end{align}
		We now bootstrap \eqref{rme} to obtain the following integral estimate:
\begin{proposition}\label{prop:propag-est1} 
{			Let the assumptions of \thmref{thm1} hold.
Suppose \propref{prop:RME-rme-red} holds}.
	Then, for every $\chi\in \cX$, 
 there exist $C=C(n,\kappa_n,\chi,v-\kappa)>0$ (with $v$ and $\kappa$ from \eqref{chi-ts} and \eqref{kappa}, resp.)  and, if $n\ge2$, $\xi^k=\xi^k(\chi)\in \cX$, $2\le k\le n$, each supported in $\supp\chi$, s.th.~{for all $s>0,\,t\ge0$,}
	\begin{align}
		&\int_0^t \chi_{s}'(r) dr  \le C \Big(s\chi_{s}(0) 
		+ \sum_{k=2}^n s^{-k+2} \ \xi^k_{s}(0) +   ts^{-n}\Big), 
		\label{propag-est31} 
	\end{align}
	where the sum should be dropped if $n=1$.
\end{proposition} 

\begin{remark} \propref{prop:propag-est1} can be reformulated in terms of expectation. 
Indeed, instead of the evolution $\chi_{s}(t)$, we could have used the expectation:
	\begin{equation}\label{br}
		 \om_t\del{\chi_{ts}} \equiv \Tr(\chi_{ts}\rho_t)
	\end{equation}
	of $\chi_{ts}$ in the state $\rho_t$ solving \eqref{1.5} and instead of \eqref{dt-Heis}, used the relation 
	\begin{align}
		{d\over{dt}}\om_t\del{\chi_{ts}} =&\om_t\del{D\chi_{st}}.\label{2.1'}
	\end{align}
	These two formulations are related through the identity
	\begin{align}\label{2form-rel}			\om_t\del{\chi_{ts}} =\om\del{\chi_{s}(t)}.
	\end{align}	
\end{remark}
\begin{proof}[Proof of \propref{prop:propag-est1}]

	For each fixed $s$, {integrating formula \eqref{3.7'} 
	with $b(t)\equiv \chi_{ts}$ in $t$ }
	gives
	\begin{align} \label{eq-basic}  
		\chi_s(t)-\int_0^t \beta_r(d\chi_{sr})\,dr= \chi_s(0).
	\end{align}
	We  apply  inequality \eqref{rme} to the second term on the l.h.s. of \eqref{eq-basic} to obtain 
	\begin{align} \label{propag-est2} 
		&\chi_s(t)+(v-\kappa)s^{-1}\int_0^t\chi'_s(r)\, dr\notag\\
		\le&  \chi_s(0)+C\del{ \sum_{k=2}^ns^{-k}\int_0^t (\xi^k)'_s(r)\,dr + t s^{-(n+1)}},
	\end{align}
	where $C=C(n,\kappa_n,\chi)>0$ and  the second term in the r.h.s. is dropped for $n=1$.  Since $\chi_s(t)\ge0$  due to the positive-preserving property of $\beta_t$ (c.f.~\eqref{beta-def}), $\kappa < v$ and $s>0$, inequality \eqref{propag-est2} implies, after dropping $\chi_s(t)$ and multiplying both sides by $s(v-\kappa)^{-1}\ge0$, that
	\begin{align}
		\int_0^t \chi'_s(r)\, dr 
		\le C 
		\del{s\chi_s(0) +   \sum_{k=2}^ns^{-k+1}\int_0^t (\xi^k)'_s(r)\,dr+ t s^{-n}}, \label{propag-est3} 
	\end{align}
	where the second term in the r.h.s.~is dropped for $n=1$. Note that from this point onward, the constant $C>0$ depends also on $v-\kappa$.
	

	If $n=1$, then \eqref{propag-est3} gives \eqref{propag-est31}. If $n\ge2$, applying \eqref{propag-est3} to the term $\int_0^t (\xi^k)'_s(r)\,dr$ for $k=2$, 
	 we obtain
	\begin{align}
		\int_0^t\chi'_s(r)\, dr \, 
		\le C \Big( \, &
		s \chi_s(0)+ \xi^2_s(0) + \sum_{k=3}^n s^{-k+1} \int_0^t (\eta^k)'_s(r)\,dr + \,{ t s^{-n}}\Big), \label{propag-est33} 
	\end{align}
	where the third term in the r.h.s.~is dropped for $n=2$, and $\eta^k=\eta^k(\xi^2,\xi^k)\in\cX, k=3,\ldots, n$. Bootstrapping this procedure, we arrive at \eqref{propag-est31}.
\end{proof}

\begin{proof}[Proof of \thmref{thm:msb-cond}]
To fix ideas, we take $t\ge0$ within this proof.  The case $t\le0$ follows from time reflection.

Fix $\chi\in\cX$ \,{with $\norm{\chi}_{L^\infty}=1$} and consider \eqref{propag-est2}.        Retaining the first term in the l.h.s.~of \eqref{propag-est2} and dropping the second one, which is non-negative since $\chi'\ge0$ and $v>\kappa$ (see \eqref{chi-ts}), we obtain
	\begin{align} \label{propag-est2bis}
        \chi_{s}(t)  \le \,\chi_{s}(0)+C\del{ \sum_{k=2}^ns^{-k}\int_0^t (\xi^k)'_{s}(r) dr + t s^{-(n+1)}},                
\end{align}
Applying \eqref{propag-est31} 
to the second term on the r.h.s.~and using property \ref{X2},
 we deduce that 				
\begin{align} \label{propag-est4} 
	\chi_{s}(t)  \le	\chi_{s}(0)+Cs^{-1}\xi_s(0)+Cs^{-n}, \end{align}
for some fixed $\xi\in\cX,\,C>0$ and all  $s> t$.

 Let $s=\eta/c$ and $\eta>  ct$. We first consider the r.h.s. of \eqref{propag-est4}.
Using $\chi_{s}(0)\equiv \chi_{0 s}$ in \eqref{chi-0s-est} and 
noting $\supp\xi^k\subset \supp\chi$ for each $k$, we find that
\begin{align}
	&\chi_{s}(0) +Cs^{-1}\xi_s(0)\le (1+Cs^{-1})\chi^\sharp_{X^\cp}.\label{2.13'}
\end{align}
By \eqref{2.13'} and property \eqref{dG-pos}, we have 
\begin{equation}
	\label{hatchi-0s-est} 
	\widehat{\chi_s(0)}+ Cs^{-1}         \widehat{\xi_s(0)}\le(1+Cs^{-1}) N_{X^c}.
\end{equation}

 Next, consider the l.h.s. of \eqref{propag-est4}.
Applying $\beta_t$ to 
\eqref{chi-ts-est}, we find that 
\begin{align}
& \beta_t(\chi^\sharp_{X_\eta^\cp})  \le \chi_{s}(t), 	\label{2.14'}			 \end{align} 
We show in Appendix \ref{sec:A}, \lemref{lemA.2}, that for every function $f$ on $\Lam$ and  $\hat f\equiv \dG(f)$, 
\begin{equation}
	\label{2.30}\al_t(\hat f)=\dG (\beta_t(f)).
\end{equation}
This, together with \eqref{2.14'}, yields
\begin{equation}
	\label{hatchi-ts-est}\al_t(N_{X_\eta^\cp}) \le  \widehat{\chi_s(t)}.
\end{equation}

Finally, 
combining estimates \eqref{propag-est4}, \eqref{hatchi-0s-est}, \eqref{hatchi-ts-est} and recalling  the assumption $\eta\ge1$, we conclude that for all $t<s=\eta/c$,
\begin{align*}
	\al_t(N_{X_\eta^\cp})&\le 
	   C(N_{X^\cp} +\eta^{-n} N),
\end{align*}
which is \eqref{MVE}. This completes the proof of \thmref{thm:msb-cond}.
\end{proof}

{\thmref{thm1} follows from  
\propref{prop:RME-rme-red}, \thmref{thm:msb-cond}, and the following lemma, proved in \secref{sec:pfmsb-cond}:  
\begin{lemma}\label{lem:chi-ts-prop}
	Let $v\in(c,\kappa)$, $s=\eta/c$, and $\delta=c-v$ in definitions \eqref{chi-ts}--\eqref{F}. Then  \eqref{chi-0s-est}--\eqref{chi-ts-est} hold for the family \eqref{chi-ts} with $\norm{\chi}_{L^\infty}=1$.
\end{lemma}}

This completes the proof of \thmref{thm1}, modulo the proofs of \propref{prop:RME-rme-red} and \lemref{lem:chi-ts-prop}, given in the next section. \qed

\section{Proofs of \propref{prop:RME-rme-red} and \lemref{lem:chi-ts-prop} }\label{sec:pfRME}
In this section, we prove the $1$-particle recursive monotonicity estimate, \propref{prop:RME-rme-red}, and the geometric  estimates \eqref{chi-0s-est}--\eqref{chi-ts-est}. 

	\subsection{Proof of \propref{prop:RME-rme-red}}\label{sec:3.1}
Recall the definition of the operators $\chi_{ts}$ in \eqref{chi-ts}.
To begin with, we prove the following lemma:
\begin{lemma}
		\label{lem5.1}
Suppose \eqref{k-cond} holds with some $n\ge1$.
			Then, for every $\chi\in\cX$,  there exist $\xi^k=\xi^k(\chi)\in\cX,\,k=2,\ldots, n$ (dropped if $n=1$), each supported in $\supp \chi$, and some $C=C(n,\kappa_n,\chi)>0$  s.th.~for all $t\in\Rb,\,s>0$,
		\begin{equation}
			\label{4.4}
			L\chi_{ts}\le s^{-1}\kappa\chi_{ts}'+ C\del{\sum_{k=2}^{n}s^{-k}(\xi^k)'_{ts}+ s^{-(n+1)}},
		\end{equation}
		where $L=i\sbr{h,\cdot}$  and $\kappa$ is as in \eqref{kappa}. 
		(The sum in the r.h.s. is dropped if $n=1$.)
	\end{lemma}
	\begin{proof}
		 Throughout the proof we fix $t$ and write $\chi_s\equiv \chi_{ts}$. \DETAILS{All estimates below are independent of $t$.} Since  $\chi'\ge0$ for $\chi\in\cX$,  expansion \eqref{4.10} with $A=ih, \,\Phi=d_X$ yields (see \corref{corC2}):
\begin{align}
	\label{4.11}
	\begin{aligned}
		L\chi_{s}=&s^{-1}\sqrt{\chi'_{s}}(Ld_X)\sqrt{\chi'_s}
		\\
		&+ \sum_{k=2}^ns^{-k}\sum_{m=1}^{N_k}g^{(m)}(s)v^{(m)}_{s} B_kv^{(m)}_{s}+s^{-(n+1)}R(s),
	\end{aligned}
\end{align}
where the sum in the second line is dropped for $n=1$. For  $n\ge2$, $1\le k \le n$,  $1\le m\le N_k$, the functions $v^{(m)}$ are piece-wise smooth and satisfy
\begin{align}
	\label{4.13}
	&\supp v^{(m)}\subset \supp\chi',\quad \norm{v^{(m)}}_{L^\infty}\le C(\chi),
\end{align}
$g^{(m)}(s)$ are piece-wise constant and take values in $\pm1$, and  $B_k=\ad{k}{d_X}{ih}$. Furthermore, by condition \eqref{k-cond},  \lemref{lemA.1}, and the remainder estimate \eqref{C4}, the operators $B_k$ and $R(s)$ are bounded on $\hf$, satisfying
\begin{align}
	\label{4.12}
	& \norm{Ld_X}\le\kappa\equiv\kappa_0,\quad \norm{B_k}\le \kappa_{k-1},\quad \norm{R(s)}\le C(n,\chi)\kappa_n,
\end{align}
with $\kappa_p$'s given in \eqref{k-cond}.

Next, adding the adjoint to both sides of \eqref{4.11}, using the self-adjointness of the first term on both sides, and then dividing the result by two,  we find
		\begin{align}\label{4.14}	
				L\chi_{s}=s^{-1}\sqrt{\chi'_{s}}(Ld_X)\sqrt{\chi'_{s}}&+\frac12\sum_{k=1}^ns^{-k}\sum_{m=1}^{N_k}g^{(m)}(s)v^{(m)}_{s} \del{B_k+B_k^*}\notag v^{(m)}_{s}\\&+\frac12 s^{-(n+1)}\del{R(s)+R(s)^*}.
		\end{align}
We can now derive an operator inequality from expansions \eqref{4.14} and uniform estimates \eqref{4.12} as 
		\begin{equation}
			\label{4.25}
			L\chi_{s} \le \kappa\chi'_{s}+C\del{\sum_{k=2}^ns^{-k} (U^k_{s})^2 +s^{-(n+1)}},
		\end{equation}
		where the sum in the r.h.s.~of  \eqref{4.25} is dropped for $n=1$ and $	C=C(n,\kappa_n,\chi)>0.$
		  For $n\ge2$,  each $U^k\in C_c^\infty$ and  is supported in $\supp \chi'$.
		
		Lastly, 
		in view of  property \ref{X1},  we find that for each $2\le k\le n$, there exists $\xi^k\in\cX$ s.th.
		$(U^k_{s})^2\le(\xi^k)'_{s}.$
		Plugging this back to \eqref{4.25} and substituting back $\chi_s\equiv \chi_{ts}$ etc.~yields \eqref{4.4}. This completes the proof.
	\end{proof}

\begin{proof}[Proof of \propref{prop:RME-rme-red}]
		We compute
	\begin{align} \label{eq:deriv}
		{\partial\over{\partial t}} \chi_{ts}=-s^{-1}v \,  \chi'_{ts}.
	\end{align}
	By \eqref{4.4}, we find 
	\begin{align*}
		L\chi_{ts} &\le  \kappa s^{-1} \chi'_{ts} +  C\del{\sum_{k=2}^ns^{-k}(\xi^k)'_{ts}+s^{-(n+1)}},
	\end{align*}
	where $C=C(n,\kappa_n,\chi)$ and the second term in the r.h.s. is dropped for $n=1$. 
	This, together with \eqref{eq:deriv} and definition \eqref{Heis-der}, implies \eqref{rme}.
\end{proof}

\subsection{Proof of \lemref{lem:chi-ts-prop}}\label{sec:pfmsb-cond}

%



First, by \eqref{F}, we have  $\supp \chi\subset (0,\infty)$, and therefore $\supp \chi\big(\frac{\cdot}{s}\big) \subset (0,\infty)$ for any $s>0$. This implies 
\begin{equation}\label{3122}
	\chi_{0s}\equiv\chi\del{\frac{d_X}{s}}\le\theta(d_X)\equiv \chi^\sharp_{X^\cp},
\end{equation}
	where 	$\theta:\Rb\to\Rb$ is the characteristic function of $\Rb_{>0}$ (see \figref{fig:f0}).
By these facts, we conclude 
\eqref{chi-0s-est}.

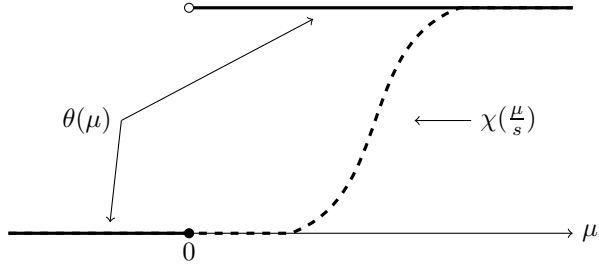
\begin{figure}[H]
	\centering
	\begin{tikzpicture}[scale=3]
		\draw [->] (-.5,0)--(2,0);
		\node [right] at (2,0) {$\mu$};
		\node [below] at (.3,0) {$0$};
		\draw [fill] (.3,0) circle [radius=0.02];
		
		

		\draw [very thick] (-.5,0)--(.3,0);
		\draw [very thick] (.3,1)--(2,1);				
		\filldraw [fill=white] (.3,1) circle [radius=0.02];
		
		\draw [dashed, very thick] (-.5,0)--(.75,0) [out=20, in=-160] to (1.5,1)--(2,1);

		\draw [->] (1.55,.5)--(1.3,.5);
		\node [right] at (1.55,.5) {$\chi(\tfrac\mu s)$};
		
		\draw [->] (0,.5)--(.85,.95);
		\draw [->] (0,.5)--(-.05,.05);
		\node [left] at (0,.5) {$\theta(\mu)$};
		
	\end{tikzpicture}
	\caption{Schematic diagram illustrating \eqref{3122}}
	\label{fig:f0}
\end{figure}

Next, again by the definition of $ \cX $, we have $\chi(\frac{\mu -vt}{s})\equiv 1$ for all $\mu\ge v\abs{t}+(c-v)s$ by setting $\delta=c-v>0$. Now, we choose 
$s=\eta/c$. Then, for all $\abs{t}< \eta/c$ and $v<c$, we have $\chi(\frac{\mu-vt}{s})\equiv 1$ for  $\mu\ge \eta$. 
This implies the estimate \begin{equation}\label{3132}
	\chi((\mu-vt)/s)\ge \theta(\mu-\eta),
\end{equation}see \figref{fig:f}.
\begin{figure}[H]
	\centering
	\begin{tikzpicture}[scale=3]
		\draw [->] (-.5,0)--(2,0);
		\node [right] at (2,0) {$\mu$};
		
		\node [below] at (0,0) {$0$};
		\draw [fill] (0,0) circle [radius=0.02];
		
		\node [below] at (1.5,0) {$\eta$};
		\draw [fill] (1.5,0) circle [radius=0.02];

		\draw [very thick] (-.5,0)--(1.5,0);
		\draw [very thick] (1.5,1)--(2,1);
		\draw [dashed, very thick] (-.5,0)--(.1,0) [out=20, in=-160] to (.65,1)--(2,1);
		\filldraw [fill=white] (1.5,1) circle [radius=0.02];

		\draw [->] (-.1,.5)--(.3,.5);
		\node [left] at (-.1,.5) {$\chi(\tfrac{\mu-vt}{s})$};
		
		\draw [->] (2,.5)--(1.75,.95);
		\draw [->] (2,.5)--(1.25,.05);
		\node [right] at (2,.5) {$\theta(\mu-\eta)$};
		
	\end{tikzpicture}
	\caption{Schematic diagram illustrating \eqref{3132}.}
	\label{fig:f}
\end{figure}
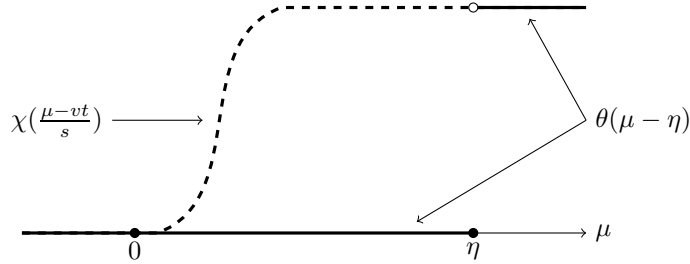
Since $X_\eta^\cp= \Set{d_X(x)>\eta}$, we have
$
	\chi^\sharp_{X_\eta^\cp}=\theta(d_X-\eta).
$
This, together with \eqref{3132}, implies
\eqref{chi-ts-est}.
This completes the proof of \lemref{lem:chi-ts-prop}.
\qed

\section{Main ideas of the proof of \thmref{thm2}}\label{sec:4.1}
Recall the notations  $X^\cp:=\Lam\setminus X$, $X_\xi\equiv \Set{x\in\Lam:d_X(x)\le \xi}$ for $\xi\ge0$ (see \figref{fig:Xxi}), and that an observable (i.e.~bounded operator) $A$ is said to be localized in $X\subset\Lambda$ if \eqref{A-loc} holds, 
written as $\supp A\subset X$. 
Next, we use the notation
\begin{equation}\label{Axi}
	A_s^\xi\equiv \al_s^{X_\xi}(A)=e^{isH_{X_\xi}}Ae^{-isH_{X_\xi}},
\end{equation} 
 where $H_Y,\,Y\subset \Lam$ is the  Hamiltonian defined by \eqref{1.1}  with $Y$ in place of $\Lam$. By definition \eqref{A-loc}, we see that  if $\supp A\subset X$, then $\supp A_s^\xi\subset X_\xi$ for all $s\in\Rb,\,\xi\ge0$.

Let $A_t=\al_t(A)$ be the full evolution \eqref{1.5}.
By the fundamental theorem of calculus, we have$$	A_t-A_t^\xi=\int_0^t \di_r\al_r(\al_{t-r}^{X_\xi}(A))\,dr.$$	 Using identity \eqref{dt-Heis} for $\al_r$ and $\al_{t-r}^{X_\xi}$ in the integrand above, as well as the fact that $\al_{t-r}^{X_\xi}([H_{X_\xi},A])=[H_{X_\xi},\al_{t-r}^{X_\xi}(A)]$, we find 
\begin{equation}\label{410}
	A_t-A_t^\xi=\int_0^t\al_r(i[R',A_{t-r}^\xi])\,dr,
\end{equation}
where $R':=H-H_{X_\xi}$. Since $A_s^\xi$ is localized in $X_\xi$, only terms in $R'$ which connect $X_\xi$ and $X_\xi^\cp$ contribute to $[R',A_{t-r}^\xi]$ (see \figref{fig:splitting}). 

\begin{figure}[H]
	\centering
	\begin{tikzpicture}[scale=1.2]
		\draw  plot[scale=.5,smooth, tension=.7] coordinates {(-3,0.5) (-2.5,2.5) (-.5,3.5) (1.5,3) (3,3.5) (4,2.5) (4,0.5) (2,-2) (0,-1.5) (-2.5,-2) (-3,0.5)};
		
		
		\node [right] at (-.5,.5) {$X_\xi$};
		
		\draw [->] (4,.5)--(1.3,.5);
		\node [right] at (4,.5) {$\text{contributing to }H_{X_\xi}$};
		
		\draw [thick] (1,.7) -- (1.4,.3);
		\draw [fill] (1,.7) circle [radius=0.05];
		\draw [fill] (1.4,.3) circle [radius=0.05];
		
		\draw [->] (4,0)--(2.3,0);
		\node [right] at (4,0) {$\text{contributing to }H_{X_\xi^\cp}$};
		
		\draw [thick] (2.05,.-.2) -- (2.4,.3);
		\draw [fill] (2.05,.-.2) circle [radius=0.05];
		\draw [fill] (2.4,.3) circle [radius=0.05];
		
		\draw [->] (4,-0.5)--(1.2,-0.5);
		\node [right] at (4,-0.5) {$\text{connecting ${X_\xi}$ and ${X_\xi^\cp}$}$};
		
		\draw [thick] (1,-0.3) -- (1.3,-1);
		\draw [fill] (1,-0.3) circle [radius=0.05];
		\draw [fill] ((1.3,-1) circle [radius=0.05];

	\end{tikzpicture}
	\caption{Schematic diagram illustrating the splitting of $H$.}
	\label{fig:splitting}
\end{figure}
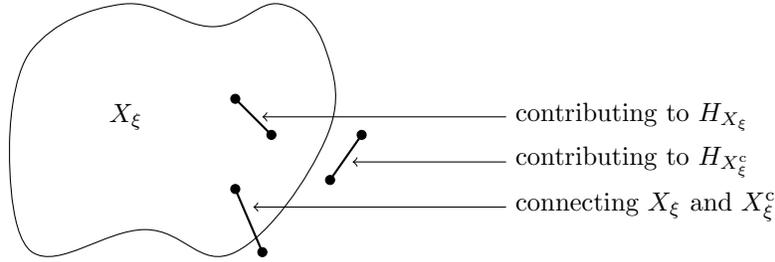

Let $s:=t-r$. Assuming first that $h$ and $v$ are finite-range, we see that the commutator $i[R',A_{s}^\xi]$ is localized near the boundary $\di X_\xi$. 
Considering for simplicity the \textit{Hubbard model}, i.e. $v_{xy}=\lam \delta_{xy}$ for some $\lam\in \R$,  and assuming $A$ and therefore $A_s^\xi$ are self-adjoint, $i[R',A_{s}^\xi]$ can be bounded, in essence, as \begin{align}\label{bdryEst}
 i[R',A_{s}^\xi]\le C\norm{A} N_{\di X_\xi}.
 \end{align}

Next, we take $X$ so that $X^\cp$ is `bounded', i.e.~independent of $\Lam$ (see Figure \ref{fig:Yxi} below) and set $Y:=X_\xi^\cp$, so that $X^\cp=Y_\xi$.  Then MVB \eqref{MVE} gives the `incoming' light cone estimate, for $r\le \xi/c,\ c>2\kappa$,  
\begin{equation}\label{MVEext}
\al_r (N_{Y}) \le C	(N_{Y_\xi} +\xi^{-n} N).\end{equation}
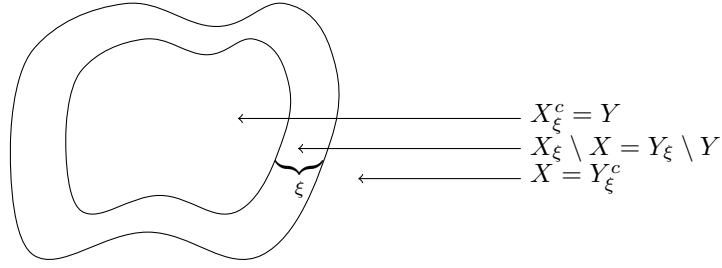
\begin{figure} [H]
	\centering
	\begin{tikzpicture}[scale=.8]
		\draw  plot[scale=.52,smooth, tension=.7] coordinates {(-3,0.5) (-2.5,2.5) (-.5,3.5) (1.5,3) (3,3.5) (4,2.5) (4,0.5) (2.5,-2) (0,-1.5) (-2.5,-2) (-3,0.5)};
		
		\draw  plot[shift={(-0.2,-0.25)}, scale=.76,smooth, tension=.7] coordinates {(-3,0.5) (-2.5,2.5) (-.5,3.5) (1.5,3) (3,3.5) (4,2.5) (4,0.5) (2.5,-2) (0,-1.5) (-2.5,-2) (-3,0.5)};
		
		\draw [->] (6,.5)--(1.3,.5);
		\node [right] at (6,.5) {$X_\xi^c=Y$};
		
		\draw [->] (6,0)--(2.3,0);
		\node [right] at (6,0) {$X_\xi\setminus X=Y_\xi\setminus Y$};	
	
		\draw [->] (6,-.5)--(3.3,-.5);
		\node [right] at (6,-.5) {$X=Y_\xi^c$};	
		
		\node [below] at (2.32,0) {$\underbrace{}_{\xi}$};
	\end{tikzpicture}
	\caption{Schematic diagram illustrating $Y, Y_\xi, Y_\xi^c$.}
	\label{fig:Yxi}
\end{figure}

Now, by the MVB \eqref{MVEext}, for any $r\le \xi/c,\ c>2\kappa$, 
  we have  \begin{equation}
		\label{MVE-pXxi} \al_r(N_{\di X_\xi})
		\le {C	(N_{(\di X_\xi)_\xi} +\xi^{-n} N)}.
	\end{equation} 
 This, together with \eqref{410}, \eqref{bdryEst} and the observation $(\di X_\xi)_\xi=X_{2\xi\setminus X}$, yields
 \begin{align}\label{lcae-gen-Hub'}
			\abs{\om\big(	\al_t(A)-\al_t^{X_\xi}(A)\big)}\le  C\abs{t}\Norm{A} 
			 \del{\om( N_{X_{2\xi}\setminus X} )+\xi^{-n} \om(N)},
		\end{align} 
This, together with \eqref{NNX}, gives  \eqref{lcae} in the finite-range case.

{For $h$ and $v$ of infinite-range, we  refine the argument presented above. Let $X_{a,b}=X_b\setminus X_a$ for $b>a\ge0$. To estimate $i[R',A_s^\xi]$ in \eqref{410}, we 
split the annulus $X_{0,2\xi}$ into four annuli, say
\begin{equation}\label{55s}
		X_{0,\frac34\xi},\quad X_{\frac34\xi,\xi}, \quad X_{\xi,\frac54\xi,},\quad X_{\frac54\xi,2\xi}.
\end{equation}
In the second and the third annuli, we use the MVB from \thmref{thm1} and in the first and the fourth ones, the decay properties of $h_{xy}$ and $v_{xy}$ as $\abs{x-y}\to\infty$. 
}See Appendix \ref{sec:4.11} for details.

 \section{Proof of \thmref{thm3}}\label{sec:6}
We introduce the remainder term for \eqref{lcae}:
	\begin{align}\label{Rem-def}
		\Rem_t(A):=\al_t(A)-\al_t^{X_\xi}(A).
	\end{align}
Since $A$ (and therefore  the  evolution $\al_t^{X_\xi}(A)$) and $B$ are respectively localized in $X_\xi$ and
	$Y\subset X_{2\xi}^\cp$,   $\al_t^{X_\xi}(A)$ and $B$ commute, yielding
	\begin{align}
		[\al_t(A), B] &=  [ \Rem_t(A), B]. \label{At-B-com}
	\end{align} 
	Next, we use Theorem \ref{thm2-gen}, which implies that for $c>2\kappa$, there exists $C=C(n,\kappa_n,\nu_n,c)>0$ s.th.~for  $\abs{t}<\xi/c$, 
\begin{equation}\label{62'}
		\begin{aligned}
&\abs{\om \del{B\,\Rem_t(A)  }}\le   C\abs{t} \Norm{A}\Norm{B} \del{\om \del{N_{X_{2\xi}\setminus X}N}+\xi^{-n}\om(N^2)}.
	\end{aligned}
\end{equation}
Since $\abs{\om(\Rem_t(A)B)}=\abs{\om(B^*(\Rem_t(A))^*) }$ and $(\Rem_t(A))^*=\Rem_t(A^*)$ (see \eqref{Rem-def}), replacing $A,\,B$ in \eqref{62'} by $A^*,\,B^*$ yields the same estimate on $\abs{\om(\Rem_t(A)B)}$, namely  $\abs{\om(\Rem_t(A)\,B)}\le $r.h.s.~of \eqref{62'} for all $\abs{t}<\xi/c$.
Then the desired estimate \eqref{LRB} follows from the triangle inequality   $\abs{\om \del{[\Rem_t(A), B]}}\le		\abs{\om \del{\Rem_t (A)B }}+\abs{\om \del{B\,\Rem_t(A)  }}$, assumption \eqref{emptyshellcond} and equality \eqref{NNX}. \qed

\section{Proof of \thmref{thm5}}\label{sec:pfthm8.2}

To fix ideas, we let $t\ge0$.  We write $C_{AB}\equiv C\norm{A}\norm{B}\om(N_Z^2)$ with $C>0$ independent of $A,\,B, N, \Lam$. Then \eqref{221} becomes 
\begin{equation}\label{8.1}
	\abs{\om_t^c(A,B)}\le C_{AB} (d_{XY}^Z/(3\l))^{1-n}
\end{equation}
for any two bounded operators $A,\,B$ localized in $X,\,Y$ with $d_{XY}^Z\ge 3\l$. 
{(For $0<d_{XY}^Z< 3\l$, \eqref{8.1} holds trivially.)}

To prove \eqref{8.1}, we use  the equality $\om_t^c(A,B)=\om^c(A_t,B_t)$ and write $A_t=A_t^\xi+\Rem_t(A)$ with $\xi:=d_{XY}^Z/3$ (see \eqref{Rem-def}) and the same for $B_t$. This way  we arrive at 
\begin{align}
	\om_t^c(A,B)=&\om^c(A_t,B_t)=\om^c(A_t^\xi, B_t^\xi)\notag\\
	&\quad +\om(\Rem_t(A)B_t)+\om(A_t\Rem_t(B))\notag\\&\quad +\om(\Rem_t(A)\Rem_t(B))\notag\\
	&\quad + \om(\Rem_t(A))\om(B_t)+\om(\Rem_t(B))\om(A_t)\notag\\ &\quad +\om(\Rem_t(A))\om(\Rem_t(B)).\label{8.1s}
\end{align}
Since $A_t^\xi$ and $B_t^\xi$ are localized in two disjoint sets $X_\xi$ and $Y_\xi$ at the distance 
$d^Z_{X_\xi,Y_\xi}=d_{XY}^Z-2\xi=d_{XY}^Z/3\ge \l,$ then,  by the $\WC(Z^\cp, \l)$ assumption on $\om$, the leading term is bounded as
\[\abs{\om^c(A_t^\xi B_t^\xi)}\le C^1_{AB}(d_{X_\xi Y_\xi}^Z/\l)^{1-n}\le C^1_{AB}(d_{XY}^Z/3\l)^{1-n},\] 
with $C_1>0$ as in \eqref{221}.

{For the $6$ trailing terms in the r.h.s.~of \eqref{8.1s}, we use \eqref{lcae-gen} and similar estimates with the roles of $A$ and $B$ interchanged, together with \eqref{2.32s},  \eqref{emptyshellcond'}, and \eqref{NNX}.
Since $X_{2\xi}\equiv X_{2d_{XY}^Z/3}$, $Y_{2\xi}\equiv Y_{2d_{XY}^Z/3}\subset Z^\cp$, we have $N_{X_{2\xi}},$ $N_{Y_{2\xi}}\le N_{Z^\cp}$. Hence, due to \eqref{2.32s}, the leading term in the r.h.s. of \eqref{lcae-gen} drops out. Thus,  these $6$ terms can be bounded by $C^2_{AB}\xi^{1-n}=C^2_{AB}(d_{XY}^Z/3)^{1-n}$ uniformly  for all
$t<\xi/3\kappa$.}


In conclusion, since $d_{XY}^Z\ge3\l$ and therefore $\xi\ge \l$, we find
\begin{equation}\label{81}
	\abs{\om_t^c(A,B)}\le C_{AB}(d_{XY}^Z/(3\l))^{1-n},
\end{equation}
for some $C=C(n,C_1,C_2)>0$ and all $t<\l/c$. 
We conclude the claim from here. 
\qed

\section{Proof of \thmref{thm6}}\label{sec:pfthm6}

Let $\SD\equiv \SD(t,r)$ as defined in \eqref{CommEstPhys'}.	The fundamental theorem of calculus yields
\begin{align} \label{CommEstPhys''} \SD=\int_0^r\tr\big[A\al_{t}'(\tau_{s}( i[\rho, B]))\big]\,ds. 
\end{align}
Since $\tau_{s}( i[\rho, B])= i[\tau_{s}(\rho), \tau_{s}(B)]$ and $\tau_{s}(B)=B$, moving $\al_{t}'$ from the state to the observable $A$ in eq.~\eqref{CommEstPhys''} gives 
	\begin{align} \label{CommEstPhys}\SD&=\int_0^rds \tr\big[\al_{t}(A) i[\tau_{s}(\rho), B]\big]\notag\\
		&=\int_0^rds \om_s\big( i[B,\al_{t}(A)]\big),\end{align}
	where $\om_s:=\Tr((\cdot)\tau_s(\rho))$ is the evolution generated by $B$. 
	\eqref{CommEstPhys} implies the upper bound
	\begin{equation}\label{232}
		\abs{\SD}\le r \sup_{0\le s\le r}\abs{ \om_s\big( [B,\al_{t}(A)]\big)}. 
	\end{equation}
%

{Let $\xi':=\xi/2\ge1,\,c':=c/2>2\kappa$, and $N_{\g,\xi}:= N_{ X_{(1-\g)\xi,(1+\g)\xi}}$ (c.f.~\eqref{N1def}). Note that $A$ and $B$ are localized  in $\cB_{X_{\xi'}}$ and $\cB_Y$, respectively, with $\dist (X_{\xi'},Y)\ge \xi'$.  Hence, by estimate \eqref{lcae-gen} (which, importantly,  is independent of state $\om$) and the relation $X_{\xi/2,3\xi/2}=(X_{\xi'})_{0,2\xi'}$, we have}
	\begin{align}
		\abs{\om_s\del{[ B, \al_t(A)] }}  
		\le C\abs{t} \Norm{A}\Norm{B} \del{ \om_s\del{N_{1/2,\xi}N}+\xi^{-n}\om_s\del{N^{2}}},\label{160}
	\end{align}
	for some  $C=C(n,\kappa_n,\nu_n,c)>0$ and all $\abs{t}<\xi'/c'$. To bound the r.h.s.~of \eqref{160}, we apply  \corref{cor3} to the evolution $\om_s$ generated by $B$ to find
\begin{equation}\label{161}
	\begin{aligned}
		\om_s\del{N_{X_{\xi/2}^\cp}N}
		\le& C \del{\om\del{N_{X^\cp}N}+\xi^{-n}\om(N^{2})},
	\end{aligned} 
\end{equation}
for some $C>0$ and all $0\le s<\xi'/c'$, and use that $\om_s(N^p)\equiv \om(N^p)$. 	Now, plugging  \eqref{160}--\eqref{161} back to \eqref{232},  using the assumption $s\le r<\xi/c=\xi'/c'$, taking $\om$ satisfying \eqref{emptyshellcond} and therefore  $\om(N_{X^\cp}N)=0$ by \eqref{2.32s}, and using relations \eqref{NNX}, we arrive at the desired estimate \eqref{7.2}.  This completes the proof. 
\qed

\section{Proof of \thmref{thm:qst}}\label{sec:pfthm:qst}

	We apply \thmref{thm2} so that $U_t^\xi\equiv\al_t^{X_\xi}(U)$  is localized in $X_\xi\subset Y^\cp$ (see \eqref{alloc}). Consequently,  conjugation by $U_t^\xi$ does not affect the partial trace $\mathrm{Tr}_{Y^\cp}$. This leads to 
	$$
	\begin{aligned}
		F\left( \mathrm{Tr}_{Y^\cp}(\rho_t), \mathrm{Tr}_{Y^\cp}(\rho_t^{U_t})\right)
		&=
		F\left(\mathrm{Tr}_{Y^\cp} \left( \rho_t^{U_t^\xi}\right), \mathrm{Tr}_{Y^\cp}\left( \rho_t^{U_t}\right)\right)\\
		&\geq
		F\left( \rho_t^{U_t^\xi}, \rho_t^{U_t}\right),
	\end{aligned}
	$$
	where the last line follows from  the data processing inequality for the fidelity, see \cite[Lem.~B.4]{FaRe}. 
	
	{Since $\om$ is a pure state,  $\om=\ondel{(\cdot)}$ for some $\varphi\in\cD(H)\cap \cD(N)$.  For rank-one projections $\rho_t=\abs{\varphi_t\rangle\langle\varphi_t}$ generated by the initial state $\abs{\varphi\rangle\langle\varphi}$,} we compute
	$$
	F\left( \rho_t^{U_t^\xi}, \rho_t^{U_t}\right)
	=\abs{\inn{U^\xi_t\varphi_t}{U_t\varphi_t}}.$$
	Since $U_t$ is unitary, so is $U_t^\xi$ (again, see \eqref{alloc}). 
	Writing $U_t=U_t^\xi+\Rem_t(U)$ and using $(U_t^\xi)^*U_t^\xi=1$, we arrive at 
	
	\begin{equation}\label{5.23'}
		\abs{\inn{U^\xi_t\varphi_t}{U_t\varphi_t}}
		\geq 1-  \abs{\otdel{(U_t^\xi)^*\Rem_t(U)}}.
	\end{equation}

%

	{Let $\xi':=\xi/4\ge1,\,c':=c/2>2\kappa$, and $N_{\g,\xi}:= N_{ X_{(1-\g)\xi,(1+\g)\xi}}$ (c.f.~\eqref{N1def}). We view $U$ as an observable in $\cB_{X_{3\xi/4}}$, so that $B\in\cB_Y$ with $\dist (X_{3\xi/4},Y)\ge \xi'$.  Take the pair  $(\varphi,\psi)$ in estimate \eqref{RemFormEst} to be $(U^\xi_t\varphi, \varphi)$. Then, by estimate \eqref{lcae-gen} (which, importantly,  is independent of state $\om$) and the relation $X_{3\xi/4,5\xi/4}=(X_{3\xi/4})_{0,2\xi'}$, we have}
 for all
	$\abs{t}<\xi'/c'=\xi/c$ that
{	\begin{align}\label{5.23}
				\abs{\ondel{(U_t^\xi)^*\Rem_t(U)}} 	\le& C\abs{t}\tau_{1/4}(\varphi)^{1/2}\tau_{1/4}(U_t^\xi\varphi)^{1/2},
		\\\text{where }\tau_\al(\phi):=&\phidel{N_{\al,\xi}N}+(\g\xi)^{-n} \phidel{N^{2}}\label{taualDef},
	\end{align}}
	for some $C=C(n,\kappa_n,\nu_n,c)>0$ and small $\g$ with $(1-\g)c'>2\kappa$. 
	Note that $\Norm{U}$ is dropped for unitary $U$.

	Let $\tilde \varphi_t=e^{-it H_{X_\xi}} \varphi$. 
	By \corref{cor3}, the first term in the second factor on the r.h.s. of \eqref{5.23} can be bounded as follows:
	\begin{equation}\label{8.5'}
		\begin{aligned}
			&\Undel{N_{1/4,\xi}N}=\inn{\tilde \varphi_t} {   U^* \al_{-t}^{X_\xi}( N_{1/4,\xi} N)  U\tilde \varphi_t}\le C\tau_{1/2}(U\tilde \varphi_t).
		\end{aligned}
	\end{equation}			
	which holds for fixed $C=C(n,\kappa_n,c)$ and all $\abs{t}<\xi/(4c')$ by \corref{cor3}. 
	
	Since $U\in \cB_X$ and $\supp N_{1/2,\xi}\subset X^\cp$, we have $[U,N]=[U, N_{1/2,\xi}]=0$. By this  and the relation $ U^* U =1$, the leading term in the last line of \eqref{8.5'} can be bounded as\begin{equation}\label{8.5}
		\begin{aligned}		
			\inn{\tilde \varphi_t}{ U^* N_{1/2,\xi} NU\tilde \varphi_t }=& \inn{\tilde \varphi_t}{  N_{1/2,\xi} N\tilde \varphi_t }\le C\tau_1(\varphi),
		\end{aligned}
	\end{equation}	
	which holds for fixed $C=C(n,\kappa_n,c)$ and all $\abs{t}<\xi/(2c')$. 
	We also have for all $t$ that 
	\begin{align}\label{8.6}
		\Undel{N^{2}}=&\ondel{N^2}.
	\end{align}
	Plugging \eqref{8.5'}--\eqref{8.6} back to \eqref{5.23},  we conclude that, 	for some fixed $C=C(n,\kappa_n,c)$ and all $\abs{t}<\xi/(4c')$,  
	\begin{equation}\label{5.24}
		\begin{aligned}
			&\abs{\ondel{(U_t^\xi)^*\Rem_t(U)}}\le C\abs{t}\tau_1(\varphi).
		\end{aligned}
	\end{equation}
	Plugging \eqref{5.24} back to \eqref{5.23'} and using the choice $c'=c/4$
	and the localization conditions \eqref{emptyshellcond}--\eqref{NNX} on the initial state $\om$,   we get the desired lower bound \eqref{eq:qst} for all $\abs{t}<\xi/c$. This completes the proof. 
\qed

\section{Proof of \thmref{thm:gap}}\label{sec:pfthm:gap}

	We adapt the argument of \cite{NachS}. 
%
%
	We shift $H$ in \eqref{1.1} so that the new Hamiltonian, which we still denote by $H$, has the ground state energy $0$ and so $H\ge0$. 
	
	Since $H=\oplus_n H_n$ and each $H_n:=H\upharpoonright_{\Set{N=n}},\,n=0,1,\ldots,$ is bounded, $e^{izH}=\prod_n e^{izH_n}$ is an entire operator-valued function of $z$. Consequently, 	 $$f(z):=\Ondel{B\al_z(A)}\quad (z\in\Cb),$$ with $\al_z(A)=e^{izH}Ae^{-izH}$ (c.f.~\eqref{1.5}), is well-defined and entire. Now, we claim that, for all $A\in\cB_X$, $B\in\cB_Y$, and all small $b>0$,
	\begin{equation}\label{281'}
		\abs{f(ib)}\le C_{AB} (\g^{-1}\xi^{-2}+\xi^{1-n}\om(N_X^2)).
	\end{equation}
	Here and in the remainder of the proof,  $C_{AB}=C\norm{A}\norm{B}$ with $C>0$ depending only on $n,\kappa_n,\nu_n,$ and $b$.
	Then the desired estimate \eqref{281} follows from the relation $f(0)=\Ondel{BA}$ by taking $b\to0+$.

	Now we prove \eqref{281'}. Let $\Cb^\pm:=\Set{z\in\Cb:\pm\Im z>0}$.
	Since $\Om$ is an eigenvector with eigenvalue $0$, we have $f(z)=\Ondel{Be^{izH}A}.$
	This,  together with the spectral theorem and the gap assumption on $H$, implies
	\begin{equation}\label{282}
		f(z)=\int_\g^\infty e^{iz\l}\,d\Ondel{BP_\l A}, 
	\end{equation}
	where $P_\l$ is the projection-valued spectral measure corresponding to $H$. To estimate the integral on the r.h.s.~of \eqref{282}, we pass to Riemann sums to obtain, for all $z\in \Cb^+$,  
	\begin{equation}\label{2.83'}
		\abs{f(z)}\le e^{-\Im z \g}\lim 	\sum_i\abs{\Ondel{BP_{\Delta_i} A}},
	\end{equation}
	where the sum is taken over a partition of  $[\g,\infty)$ into subintervals $\Delta_i$'s. Next, using Cauchy-Schwartz inequality, we estimate
	\begin{equation}\label{2.83''}
		\begin{aligned}
			\sum\abs{\Ondel{BP_{\Delta_i} A}}\le& \sum \norm{P_{\Delta_i} B^*\Om}\norm{P_{\Delta_i} A\Om}
			\\\le&\del{\sum\norm{P_{\Delta_i} B^*\Om}^2}^{1/2}\del{\sum\norm{P_{\Delta_i} A\Om}^2}^{1/2}\\\le& \norm{B^*\Om}\norm{A\Om},
		\end{aligned}
	\end{equation}
with the norms on the r.h.s.~taken in the Fock space $\cF$.
	Since \eqref{2.83''} is uniform in all partitions, 
	combining \eqref{2.83'}--\eqref{2.83''} yields
	\begin{equation}\label{283}
		\begin{aligned}
			\abs{f(z)}\DETAILS{\le& e^{-\Im z \g}\int_\g^\infty \abs{d {\Ondel{AP_\l B}}}
				\\}\le& \norm{A}\norm{B}e^{-\Im z \g}.
		\end{aligned}
	\end{equation}

	Next, fix $T>0$ to be chosen later.  \DETAILS{Denote by $f(z)=\Ondel{Ae^{izH}B}$ an entire extension of $f(t)$, which exists by assumption. } Since $f(z)$ is entire,
	by the Cauchy integral formula, for every $0<b< T$, we have  
	\begin{equation}\label{284}
		f(ib)=\frac{1}{2\pi i}\del{\int_{\Ga_T^+}\frac{f(z)dz}{z-ib}+\int_{-T}^T \frac{f(t)\,dt}{t-ib}},
	\end{equation}
	where $\Ga_T^+\subset \Cb^+$ denotes the semicircle with radius $T$ in the upper half-plane $\Cb^+$. Moreover, for all sufficiently small $b$, we have $\abs{z-ib}>T/2$ for all $z\in\Ga_T^+$, whence
	\begin{equation}\label{286}
		\abs{\int_{\Ga_T^+}\frac{f(z)dz}{z-ib}}\le \frac{\norm{A}\norm{B}}{T}\int_0^\pi e^{-\g T\sin\theta}\,d\theta\le \frac{C_{AB}^1}{\g T^2},
	\end{equation}
	by estimate \eqref{283}.
	This bounds the first term in the r.h.s. of \eqref{284}.
	
	To bound the second term in the r.h.s. of \eqref{284}, we take some $0<\delta<T$ to be determined later, split the interval $I_T=I_\delta\cup (I_T\setminus I_\delta)$ (where $I_a:=[-a,a]$), and write, for every $t\in I_T$, \begin{equation}\label{285}
		f(t)=\Ondel{\al_t(A)B}+\Ondel{[B, \al_t(A)]}=:\mathrm{I}(t)+\mathrm{II}(t).
	\end{equation}Then we have
	\begin{equation}\label{287}
		\begin{aligned}
			\abs{\int_{-T}^T \frac{f(t)\,dt}{t-ib}}\le& \abs{\int_{-T}^T\frac{\mathrm{I}(t)\,dt}{t-ib}}+\abs{\int_{-\delta}^\delta\frac{\mathrm{II}(t)\,dt}{t-ib}}+\abs{\int_{I_T\setminus I_\delta}\frac{\mathrm{II}(t)\,dt}{t-ib}}
			\\=:& F+G_1+G_2.
		\end{aligned}
	\end{equation}
	To bound $F$, we note that by the Cauchy–Goursat theorem, $F=\abs{\int_{\Ga_T^-}\frac{\mathrm{I}(z)\,dz}{z-ib}}$, where and $\Ga_T^-\subset \Cb^-$ denotes the semicircle with radius $T$ in the lower half-plane $\Cb^-$.
	Therefore, by the same argument as \eqref{286}, we find that $F$ satisfies the estimate 
	\begin{equation}\label{289}
		F\le \frac{C_{AB}^1}{\g T^2 }.
	\end{equation}
	To bound $G_1$, we note that since $\mathrm{II}(t)$ is analytic and vanishes at $t=0$, we have $\abs{\mathrm{II}(t)}\le C_{AB}^2\abs{t}$ for all small $t$. This implies \begin{equation}\label{2810}
		G_1\le \int_{-\delta}^\delta \frac{\abs{\mathrm{II}(t)}\,dt}{\abs{t} }\le C_{AB}^2\delta.
	\end{equation}
	To bound $G_2$, we note that by the weak LRB \eqref{LRB}, $\mathrm{II}(t)$ satisfies the uniform estimate $\abs{\mathrm{II}(t)}\le C_{AB}^3\abs{t}\xi^{-n}\om(N_X^2)$ for all real $t$ with $\abs{t}<\xi/(3\kappa)$. Hence,
	\begin{equation}\label{2811}
		G_2\le \int_{I_T\setminus I_\delta} \frac{\abs{\mathrm{II}(t)}\,dt}{\abs{t} }\le C_{AB}^3 (T-\delta)\xi^{-n}\om(N_X^2),
	\end{equation}
	{provided $T<\xi/(3\kappa)$. } Combining \eqref{287}--\eqref{2811} yields an estimate on the second term in the r.h.s. of \eqref{284}:
	\begin{equation}\label{2812}
		\abs{\int_{-T}^T \frac{f(t)\,dt}{t-ib}}\le C_{AB}(\g^{-1} T^{-2}+\delta +(T-\delta)\xi^{-n}\om(N_X^2) ).
	\end{equation}
	Finally, choosing $\delta=\xi^{-n}/(10\kappa)<T=\xi/(6\kappa)$ {(recall $\xi\ge1$)}, and plugging \eqref{286}, \eqref{2812} back to \eqref{284}, we find
	$$\abs{f(ib)}\le C_{AB}(\g^{-1}\xi^{-2}+\xi^{-n}+\xi^{1-n}\om(N_X^2) ).$$
	We conclude claim \eqref{281'} from here. This  completes the proof
\qed

\section{Proof of \thmref{thm4}}\label{sec:pfthm4}
We follow the argument in Sects.~\ref{sec:pfthm1}--\ref{sec:pfRME}.
For $\chi\in\cX$ (see \eqref{F}) and two numbers $\abs{t}<s$, define the ASTLOs
\begin{equation}\label{astlo}
	\bar\chi_{ts}:=\hat \chi_{ts}/N,
\end{equation}	
where, recall, $\hat \chi_{ts}$ is given by \eqref{PIO1}. By relation \eqref{dG-com1}, we see that $\bar\chi_{ts}$ commutes with $\bar\xi_{t's'}$ for any two functions $\chi,\xi$. Define a set of smooth cutoff functions
\begin{equation}\label{6.3}
	\begin{aligned}
		\cG\equiv \cG_{\nu,\nu'}
		:=&\Set{f\in C^\infty(\R)\left|
			\begin{aligned}
				&\supp f\subset \Rb_{\ge0},\,\supp f'\subset (\nu,\nu')\\
				&\,f^\prime\ge 0,\,\sqrt{f'}\in C^\infty(\R)
			\end{aligned}\right.
		}.
	\end{aligned}
\end{equation}
For any  $f\in\cG,\,\chi\in \cX$, we consider the two parameter family of operators 
\begin{align}	
	f_{ts}:=f(\bar\chi_{ts}).
\end{align}
We now claim that this family satisfies a recursive monotonicity estimate similar to \eqref{RMB}: 
Namely, 
there exist constant
$C>0$ and function
$\xi_k\in\cX$ s.th. ~for all $t\in\Rb,\,s>0$,
\begin{align}\label{6.5}
	D f_{ts}
	\leq & f'_{ts}\del{\frac{\kappa-v}{s}\overline{\chi'}_{ts}+\sum_{k=2}^{n}s^{-k}\overline{(\xi^k)'}_{ts}+Cs^{-(n+1)}},
\end{align}
where, recall, $D=\di_t+i[H,\cdot]$ is the Heisenberg derivative from \eqref{Heis-der}, $f'_{ts}\equiv f'(\bar\chi_{ts}),$ and $\overline{\chi'}_{ts}\equiv\hat\chi'_{ts}/N$. (The sum in the r.h.s. is dropped for $n=1$.)

Let $R_{ts}(z)=(z-\bar\chi_{ts})^{-1}$ for $\Im z\ne0$.  Since $f$ is smooth and has compactly supported derivatives, by the Helffer-Sj\"orstrand formula (see \cite[Lem.~B.2]{HunSig1}), \begin{align}\label{HSform}
	f_{ts}^{(p)}=\int R_{ts}^{p+1}(z)\,d\tilde f(z),\quad p=0,1,
\end{align}for some finite measure $d\tilde f(z)$ on $\Cb$ vanishing for $\Im z=0$. By \eqref{HSform}, together with the relations $Df_{ts}=\int D R_{ts}(z)\,d\tilde f(z)$ and $D R_{ts}=R_{ts}(D\bar\chi_{ts})R_{ts}$, we compute
\begin{align}
	\label{Hfcomm1}
	\begin{aligned}
		Df_{ts}
		=&\int R_{ts}(z)D \bar\chi_{ts} R_{ts}(z)\,d \tilde f(z).	
	\end{aligned} 
\end{align}
(One can consider \eqref{Hfcomm1} as an\textit{ integral chain rule} for the Heisenberg derivative.)
Since $[H,N]=0$, by definition \eqref{astlo}, we have $D\bar\chi_{ts}=D\hat\chi_{ts}N^{-1}$. Plugging this back to \eqref{Hfcomm1} and applying the recursive monotonicity estimate \eqref{RMB}, we find 
\begin{align}
	\label{Hfcomm2}
	Df_{ts}\le  \int R_{ts}(z)\del{-\frac{v -\kappa}{s} \overline{\chi'}_{ts}+\sum_{k=2}^{n}s^{-k}\overline{(\xi^k)'}_{ts}+ Cs^{-(n+1)}}R_{ts}(z)\,d \tilde f(z),
\end{align} 
for some $\xi_k\in\cX$ and $C=C(n,\kappa_n,\chi)>0$. (The sum in the r.h.s. is dropped for $n=1$.) Finally, using that $[R_{ts}(z),\bar{\eta}_{ts}]=0$ for $\eta=\chi',\,\xi_k'$,  estimate \eqref{Hfcomm2}, and representation formula \eqref{HSform}, we conclude claim \eqref{6.5}.

With the recursive monotonicity estimate \eqref{6.5}, we can proceed exactly as in the proof of \propref{prop:propag-est1} and the derivation of \eqref{propag-est4} to obtain
\begin{equation}\label{6.10}
	\al_t(f_{ts})\le C\del{f_{0s}+s^{-n}},
\end{equation}	
for some $C=C(n, \kappa_n,c,\nu'-\nu)>0$ and all $s>\abs{t}$. Following the same argument as in \secref{sec:pfmsb-cond}, for appropriately chosen $f\in\cG$, we obtain the estimates
\begin{equation}\label{6.11}
	f_{0s}\le  P_{\bar N_{X^\cp}\ge \nu}, \quad P_{\bar N_{X_\eta^\cp}\ge \nu'}\le f_{ts},
\end{equation}
c.f.~\eqref{chi-0s-est}--\eqref{chi-ts-est} as well as \cite[Eq.~(15)]{FLS}. By assumption \eqref{6.0} and estimates \eqref{6.10}--\eqref{6.11}, we conclude \eqref{6.3}.
\qed

\section*{Acknowledgment}

The first author is grateful to Jeremy Faupin, Marius Lemm, and Avy Soffer for enjoyable and fruitful collaborations. Both authors thank M.~Lemm for helpful remarks on the manuscript and the very useful observation \eqref{NNX}. 
The research of I.M.S. is supported in part by NSERC Grant NA7901.
J.Z.~is supported by National Key R \& D Program of China Grant 2022YFA100740, China Postdoctoral Science Foundation Grant 2024T170453, National Natural Science Foundation of China Grant 12401602, and the Shuimu Scholar program of Tsinghua University. His research was also supported in part by DNRF Grant CPH-GEOTOP-DNRF151, DAHES Fellowship Grant 2076-00006B, DFF Grant 7027-00110B, and the Carlsberg Foundation Grant CF21-0680, and NSERC Grant NA7901.  
Parts of this work were done while the second author was visiting MIT. 
Earlier version of parts of this work has appeared as a chapter in J.Z.'s PhD thesis at the University of Copenhagen.

\section*{Declarations}
\begin{itemize}
	\item Conflict of interest: The Authors have no conflicts of interest to declare that are relevant to the content of this article.
	\item Data availability: Data sharing is not applicable to this article as no datasets were generated or analysed during the current study.
\end{itemize}

		\appendix

\section{Fock spaces}\label{sec:Fock-sp}	
In this appendix we discuss some general properties of Fock spaces used in this paper, see \cite{DerGer, FrGrSchl}. 

Given a  ($1$-particle) Hilbert space $ \frak h$, one defines the Fock space (over $ \frak h$) as
\begin{equation}\label{fockdef}
	\cF\equiv 
\cF(\frak h):=\oplus_{n=0}^\infty \otimes^n \frak h,
\end{equation}
where $\otimes^n \frak h=\C$ for $n=0$, $=\frak h$ for $n=1$, and is the symmetric (or anti-symmetric) tensor product of  $\frak h$'s for $n>1$.

 For any two ($1$-particle) Hilbert spaces $\frak h_1$ and $\frak h_2$, there is a unitary  map $U\equiv U_{(\frak h_1, \frak h_2)}$ s.t. 
\begin{align}\label{Gam-factor}U: \Gamma(\frak h_1\oplus \frak h_2)\ra \Gamma(\frak h_1) \otimes \Gamma(\frak h_2).\end{align}
Let $p_i$ be the projection from $\hf_1\oplus \hf_2$ to $\hf_i$. Then the map $U$ is defined as follows
\begin{align}\label{Gam-factor'}U\big|_{\otimes^n(\hf_1\oplus \hf_2)}:=\sum_{k=0}^n \left( \begin{array}{c} n \\ k \end{array} \right)^{1/2}p_1^{\otimes (n-k)}\otimes p_2^{\otimes k},\end{align}
where $p_i^{\otimes m}$ denotes the $m$-fold tensor product $p_i\otimes \cdots\otimes p_i$.

Furthermore, the decoupling operator  $U\equiv U_{(\frak h_1, \frak h_2)}$ can also be constructed using  creation and annihilation operators by setting
\begin{equation}\label{eq:Urhodefn}
\begin{aligned}
	&U_{(\hf_1,\hf_2)}\Om :=\Om_1\otimes \Om_2 ,\\ &U_{(\frak h_1, \frak h_2)}  a^\sharp(f)=\big(a^\sharp(f_1)\otimes\one+\one\otimes a^\sharp(f_2) \big) U_{(\frak h_1, \frak h_2)},
\end{aligned}\end{equation}
for  $f= f_1\oplus f_2\in \frak h_1\oplus \frak h_2$, and using these formulae to define $U_{(\hf_1,\hf_2)}$ on an arbitrary vector in $\Ga(\hf_1\oplus\hf_2)$. Here, $\Omega_\sharp$ is the vacuum in $\cF_\sharp$ and  
$a^\sharp(f)=\sum_{x\in\Lam} a^\sharp_x f(x)$ with $a^\sharp$ standing for either  $a$ or $a^*$.

A natural example of the splitting $\frak h=\frak h_1\oplus \frak h_2$  is the splitting of the Hilbert space $\ell^2(\Lam)$ as
$$\ell^2(\Lam)=\ell^2(S) \oplus \ell^2(S^\cp)$$  
for any $S\subset\Lam$ and with $S^\cp=\Lam\setminus S$. 
Denoting the corresponding unitary map by $U_S$, we have  
\begin{align}\label{F-factor}U_S: \cF\ra \cF_S\otimes \cF_{S^\cp},\end{align}
where $\cF_S$ is the  Fock space over the $1$-particle Hilbert space $\ell^2(S)$, 
\begin{align*}
&\cF_S:=\G(\ell^2(S))\equiv \cF( \ell^2(S) ).\end{align*}

Then, an observable $A$ on the Fock space $\cF$ is supported (or localized) in $S$ in the sense of \eqref{A-loc} if and only if it is of the form  \begin{equation}\label{loc-cond}
	 U_{S}AU_{S}^*=A_S\otimes \one_{S^\cp},
\end{equation}
where $A_S$ is the restriction of $A$ on $\cF_{S}$ and similar for $ \one_{S^\cp}$.

\section{Technical estimates}\label{sec:A}
Throughout this section, let $\Lam$ be a connected subset of a lattice $\Lc\subset \Rb^d,\,d\ge1$ with grid size $\ge1$. Denote by $\hf:=\ell^2(\Lam)$ and by $\cF$ the (fermionic or bosonic) Fock space over $\hf$.
\begin{lemma}
	\label{lemA.1}
	Let $n\ge1$. Suppose $A$ is an operator acting on  $\hf$ with operator kernel (matrix) $A_{xy}$ satisfying 
	\begin{equation}
		\label{A.1}
		M:=\del{\sup _{x\in\Lam} \sum_{y\in\Lam} \abs{A_{xy}}\abs{x-y}^{n+1}}\del{\sup _{y\in\Lam} \sum_{x\in\Lam} \abs{A_{xy}}\abs{x-y}^{n+1}}<\infty.
	\end{equation}
Then for every function $f$ on $\Lam$ s.th.~for some $L>0$,
\begin{equation}\label{phicond}
	\abs{f(x)-f(x)}\le L \abs{x-y}\quad (x,y\in\Lam),
\end{equation}
we have
\begin{equation}\label{CommEst}
	\norm{\ad{k}{f}{A}}\le L^k M\quad (1\le k\le n+1).
\end{equation}
\end{lemma}
\begin{proof}
	For every $f:\Lam\to \Cb$, a simply induction shows that for all $k$,
\begin{equation}\label{A.2}
		\del{\ad{k}{f}{A}}_{xy}=A_{xy}(f(y)-f(x))^k.
\end{equation}
	This formula, together with the Schur test for matrices, implies
\begin{equation}\label{A.3}
		\norm{\ad{k}{f}{A}}^2\le \del{\sup _{x\in\Lam} \sum_{y\in\Lam} \abs{A_{xy}}\abs{f(x)-f(y)}^k} \del{\sup _{y\in\Lam} \sum_{x\in\Lam} \abs{A_{xy}}\abs{f(x)-f(y)}^k}.
\end{equation}
If $f$ satisfies \eqref{phicond}, then
\begin{align}\label{B6}
\sup _{x\in\Lam} \sum_{y\in\Lam} \abs{A_{xy}}\abs{f(x)-f(y)}^k
		\le& L^k \sup _{x\in\Lam} \sum_{y\in\Lam} \abs{A_{xy}}\abs{x-y}^{k}\le L^kM,
	\end{align}
where the last estimate follows from assumptions \eqref{A.1} and that the grid size of $\Lam$ is at least $1$. Similarly we can show that the second term in the r.h.s.~of \eqref{A.3} satisfies the same bound as \eqref{B6}. Plugging the results back to \eqref{A.3}
  completes the proof.
\end{proof}

\begin{corollary}\label{corA.2}
	Suppose $H$ in \eqref{1.1} satisfies \eqref{k-cond}. Then, for every $X\subset \Lam$ and the distance function $d_X(x)=\dist(\Set{x}, X)$, we have
	$$
	\norm{\ad{k}{d_X}{H}}\le M\quad (1\le k\le n+1).
	$$
\end{corollary}
\begin{proof}
	 Every $d_X$ is uniformly Lipschitz and  satisfies \eqref{phicond} with $L=1$. 
\end{proof}

\begin{lemma}\label{lemA.2}
	Let $\al_t$ (resp. $\beta_t$) be the many-body (resp. $1$-body) evolutions generated by $H=\dG(h)+V$ (resp. $h$). Suppose $V$ satisfies \eqref{V-cond}. Then for every function $f$ on $\Lam$ and its second quantization $\hat f$ as in \eqref{1.7}, 
	\begin{equation}
		\label{A.6}\al_t(\hat f)=\dG (\beta_t(f)).
	\end{equation}
\end{lemma}
\begin{proof}
	Without loss of generality, we take $t\ge0$ within this proof. Write $H_0:=\dG(h)$. We decompose the evolution $\al_t$ into a composition of two maps:
	\begin{align}
		\al_t =&\al_t^\intn\circ \al_t^\loc    \label{A.7.1},\\
		\al_t^\loc (A)=& e^{itH_0}A e^{-itH_0}, \label{A.7.2}\\
		\al_t^\intn (A)=& e^{itH}e^{-itH_0} A e^{itH_0}e^{-itH}.	\label{A.7.3}
	\end{align}

For every function $f:\Lam\to \Cb$ and $f_r:=\beta_r(f)$, we compute using \eqref{A.7.2} that 
\begin{align*}
	\frac1i\di_r \al_{t-r}^\loc(\dG(f_r))=&\al_{t-r}^\loc\del{-\sbr{H_0,\dG (f_r)}+\dG (\sbr{h,f_r})}.
\end{align*}
Applying \eqref{dG-com1} to the second term on the r.h.s., we see that $\di_r \al_{t-r}^\loc(\dG(f_r))=0$. Hence 
\begin{equation}
	\label{A.8}
	\dG(f_t)-\al_t^\loc(\hat f)=\int_0^t \di_r \al_{t-r}^\loc(\dG(f_r))\,dr=0,
\end{equation}
where, recall, $\hat f=\dG( f)$. 

Next, using \eqref{A.7.3}, we compute, for every observable $A$, 
\begin{equation}\label{B11}
	\frac1i \di_r \al_r^\intn(A)=\al_r^\intn\del{\sbr{\al_r^\loc (V),A}},
\end{equation}
and therefore
\begin{equation}\label{A.7}
	\begin{aligned}
		\al_t(\hat f)-\al_t^\loc(\hat f)=&  \int_0^t \di_r\del{\al_r^\intn\circ \al_t^\loc(\hat f)}\,dr\\
=&i \int_0^t \al_r^\intn\del{\sbr{\al_r^\loc (V),\al_t^\loc(\hat f)}} \,dr\\
=&i \int_0^t \al_r \del{\sbr{V,\al_{t-r}^\loc(\hat f)}} \,dr\\
=&i \int_0^t \al_r \del{\sbr{V,\dG(f_{t-r})}}=0,
	\end{aligned}
\end{equation}
where in the last line we use \eqref{A.8} and property \eqref{V-cond}. Combining \eqref{A.8}--\eqref{A.7} gives  \eqref{A.6}. 
\end{proof}

		\section{Symmetrized commutator expansion}\label{sec:B}
In this appendix, we establish the following symmetrized commutator expansion (c.f.~\cite{FLS, Breteaux_2022}):

\begin{proposition}
	\label{thm4.1} 		Let $A\in\cB(\hf)$ and $\Phi$ be a self-adjoint operator on $\hf$
	s.th.~for some $n\ge1$,
	\begin{equation}\label{4.3}
		\ad{k}{\Phi}{A}\in\cB(\hf)\quad (1\le k\le n+1).
	\end{equation}
	Then,  for every  $\chi\in C^\infty(\Rb)$ s.th.~$\chi'$ has compact support and operators \begin{equation}\label{chitsphi}
		\chi_{ts}:=\chi(s^{-1}(\Phi-vt))
	\end{equation} with $s,\,t\in\Rb$ (c.f.~\eqref{chi-ts}), we have the expansion
	\begin{equation}
		\label{4.10}
		\begin{aligned}
			\sbr{A,\chi_{ts}}=&s^{-1}\sqrt{\abs{\chi'_{ts}}}\sgn(\chi'_{ts})[A,\Phi] \sqrt{\abs{\chi'_{ts}}}\\&+\sum_{k=2}^ns^{-k}\sum_{m=1}^{N_k}v^{(m)}_{ts}g^{(m)}_{ts} \ad{k}{\Phi}{A} v^{(m)}_{ts}+s^{(n+1)}R(t,s).
		\end{aligned}
	\end{equation}
The r.h.s.~of \eqref{4.10} is dropped for $n=1$. Moreover, if $n\ge2$, \begin{enumerate}
	\item $v^{(m)}$ are piece-wise smooth functions supported in $\supp(\chi')$;
	\item  $g^{(m)}$ are piece-wise constant functions taking values in $\pm1$ on $\supp(v^{(m)})$;
	\item  $v^{(m)}_{ts},\,g^{(m)}_{ts}$ are defined as \eqref{chitsphi}; \item $R(t,s)$ is bounded for all $s,\,t,$  and  satisfies
	\begin{align}\label{C4}
		\norm{R(t,s)}\le& C\norm{\ad{n+1}{\Phi}{A}},
	\end{align}
	for some $C=C(n,\chi)>0$;
	\item $1\le N_k\le C(n)$ for some $C(n)>1$ and all $k=2,\ldots, n$.	
\end{enumerate}\end{proposition}

	Since $\chi\in\cX$ and $h$ satisfies \eqref{k-cond}, by definition \eqref{F} and  \corref{corA.2}, the hypotheses of \propref{thm4.1} are satisfied with $\chi\in \cX$,  $A=ih$ (see \eqref{1.1}),  and $\Phi=d_X$ for any $X\subset \Lam$. This gives :
	\begin{corollary}\label{corC2}
		Let $\Phi=d_X$ in \eqref{chitsphi}. Then the operators $\chi_{ts}$  satisfies \eqref{4.10} with $A=ih,\,\Phi=d_X$. 
	\end{corollary}

\begin{proof}[Proof of \propref{thm4.1}]
		Throughout the proof, we fix $t$ in \eqref{chi-ts} and the \textit{self-adjoint operator} $\Phi$ satisfying \eqref{4.3}.  Then we consider the one-parameter  family of (bounded) operators $\chi_s\equiv \chi_{ts}$, see \eqref{chitsphi}.
In the proof below, all estimates are independent of $\Phi,\,v,\,t$.

	1. Since $\chi\in C^\infty$ and $\chi'$ has compact support, the hypotheses of \cite[Lems.~B.1--2]{HunSig1} are satisfied. Hence, by the commutator expansion formula \cite[Eq.~(B.14)]{HunSig1}, we have
	\begin{align*}
		[A,  \chi_s]=&  \sum_{k=1}^ns^{- k}E^{(0)}(k,s)+s^{-(n+1)}R^{(0)}(s),\\
		E^{(0)}(k,s):=&{1\over{k!}} \chi^{(k)}_s\ad{k}{\Phi}{A}.
	\end{align*}
 with $R^{(0)}(s)$ satisfying the remainder estimate, 
	$$\norm{R^{(0)}(s)}\le C\norm{\ad{n+1}{\Phi}{A}},$$
{where $C=C(n,\,\chi)>0$ and the r.h.s. is finite by condition \eqref{4.3}.}
 We proceed to symmetrize $E^{(0)}(k,s),\,k=1,\ldots,n$  w.r.t. the functions  $G_k^{(0)}(s):= \chi^{(k)}_s$. For each $k$, let
	\begin{align*}
		&v_k^{(0)}(s)\equiv \del{v_k^{(0)}}_s :=\sqrt{\abs{G_k^{(0)}(s)}},\quad g_k^{(0)}(s)\equiv \sgn(G^{(0)}_k(s)).
	\end{align*}
	Then we have 
	\begin{equation}\label{4.6}
		\begin{aligned}
			G_k^{(0)}(s)\ad{k}{\Phi}{A} =& g_k^{(0)}(s)v_k^{(0)}(s)  \ad{k}{\Phi}{A} v_k^{(0)}(s)\\ &+ g_k^{(0)}(s)v_k^{(0)}(s) \sbr{v_k^{(0)}(s),\ad{k}{\Phi}{A}}.
		\end{aligned}
	\end{equation}
	By the assumption on $\chi$, each $v_k^{(0)}$ is a piece-wise smooth function supported $\supp\chi'$. Hence, we can again expand the commutator in the r.h.s. of \eqref{4.6} via \cite[Eq.~(B.14)]{HunSig1}. This way we obtain
	\begin{align}\label{4.7}
		[  v_k^{(0)}(s),\ad{k}{\Phi}{A}]= -\sum_{m=1}^{n-k} {s^{- m}\over{m!}}G^{(1)}_{m,k}(s)\ad{k+m}{\Phi}{A}+s^{-(n-k+1)}R^{(1)}_k(s),
	\end{align}
	where \begin{enumerate}
		\item[(a)] the sum is dropped if $k=n$;
		\item[(b)]  $\supp G^{(1)}_k(s)\subset \supp v_k^{(0)}(s)\subset \supp \chi'_s$,  with $\norm{G^{(1)}_k(s)}_{L^\infty}\le C$ for some $C=C(\chi)$ and all $k,s$;
		\item[(c)] and  $\norm{R^{(1)}_k(s)} \le C\norm{\ad{n+1}{\Phi}{A}}$ for some $C=C(n,\chi)>0$ and all $k,s$.  
	\end{enumerate}
	
	If $n=1$, then the first term in \eqref{4.7} is dropped. Hence, plugging \eqref{4.7} into \eqref{4.6}, we find
	$$[A,\chi_s]=g_1^{(0)}(s)v_1^{(0)}(s)\ad{1}{\Phi}{A}v_1^{(0)}(s)+s^{-2}(R^{(0)}(s)+R^{(1)}_1(s)).$$
	This establishes expansion  \eqref{4.10} for $n=1$.
	
	2. If $n\ge2$, then we iterate Step 1 as follows. First, plugging \eqref{4.7} into \eqref{4.6}, we find
	\begin{align*} 
		[A,  \chi_s]=&\sum_{k=1}^n s^{-k} \Bigl(S^{(1)}(k,s)+\sum_{m=1}^{n-k+1} s^{- m}E^{(1)}(k,m,s)\Bigr)\\
		&+s^{-(n+1)}\del{R^{(0)}(s)+R^{(1)}(s)},\notag\\
		S^{(1)}(k,s):=&\frac{1}{k!}g_k^{(0)}(s) v_k^{(0)}(s) \ad{k}{\Phi}{A} v_k^{(0)}(s),\\
		E^{(1)}(k,m,s):=&-{1\over{m!}}g_k^{(0)}(s)v_k^{(0)}(s)  G_{k,m}^{(1)}(s)\ad{k+m}{\Phi}{A},\label{4.3.11}\\
		R^{(1)}(s):=&\sum_{k=1}^n\frac{1}{k!}g_k^{(0)}(s)v_k^{(0)}(s) R^{(1)}_k(s).
	\end{align*}
	
	%
	
	Fix $1\le k\le n-1$.  We symmetrize each of $E^{(1)}(k,m,s),\,m=1,\ldots,n-k+1$ w.r.t. the function \begin{equation}\label{sym}
		-{1\over{m!}}g_k^{(0)}(s) v_k^{(0)}(s) G_{k,m}^{(1)}(s)
	\end{equation} in place of $G^{(0)}_k(s)$ in Step 1. This will introduce symmetrized operators $S^{(2)}(k,m,s)$, uniformly bounded operators  $E^{(2)}(k,m,l,s)$,  and   remainders  $R^{(2)}_{k,m}(s)$ as before. 
	
	3. From here one can see that  this process can be iterated for exactly $(n-1)$--times.  At the end, we obtain an expansion of $[A,  \chi_s]$ into a sum of the form 
	\begin{equation}
		\label{4.A.11}
		[A,  \chi_s]
		=\sum_{k=1}^ns^{-k}\sum_{p=1}^{n-1} \sum S^{(p)}(k_1,\ldots, k_p)+\sum_{k=0}^kR^{(k)}(s),
	\end{equation}
	where the third sum is over some  
	combinations of $k_i\ge1$ with $\sum_{i=1}^p k_i=k$. 
Each $S^{(p)},\,p=1,\ldots,n-1$ is of the form
	\begin{equation*}
		S^{(p)}\del{k_1,\ldots,k_p,t,s}=(-1)^{p-1}\frac{1}{k_p!}g^{(p-1)}_{k_1\ldots k_p}(s)v^{(p-1)}_{k_1\ldots k_p}(s)\ad{k_1+\ldots +k_p}{\Phi}{A}v^{(p-1)}_{k_1\ldots k_p}(s),
	\end{equation*}
	where the functions
	$v^{(p-1)}_{k_1\ldots k_p}\text{ are piece-wise smooth},$
	 uniformly bounded by a constant $C=C(\chi)$, and supported in $\supp \chi'$. 
	$\abs{g^{(p-1)}_{k_1\ldots k_p}(s)}$ are piece-wise constant functions, taking values in $\pm1$. This establishes expansion \eqref{4.10}. The uniform bound on  the remainder follows from corresponding uniform estimates obtained above. 
\end{proof}

\section{Proof of \thmref{thm2-gen}}\label{sec:4.11}

Within this proof we assume $t\ge 0$. The case $t\le0$ follows by the time reflection.  Recall the notations $A_t=\al_t(A)$ and $A_t^\xi=\al_t^{X_\xi}(A)$ for $\xi\ge0$, see \eqref{alloc}.

For a mixed state $\om$, we decompose $\om =\sum p_iP_{\psi^i}$, where $P_\psi$ is the rank-one projection onto $\Cb\psi$,  with  $p_i\ge0,\,\sum p_i<\infty$, and use linearity to reduce the problem to estimating $\abs{\inn{\varphi}{(A_t-A_t^\xi)\psi}}$ for appropriate $\varphi$ and $\psi$. The latter is done in Step 4 at the end of the proof.

1. 
We fix an operator $A\in \cB_X$ and define the remainder operator $$\Rem_t\equiv \Rem_t(A)=A_t-A_t^\xi,$$  c.f.~\eqref{Rem-def}, 
as well as the ($X_\xi$-$X_\xi^\cp$)-coupling operator (c.f.~$R'$ in \eqref{410}) \[R=H-H_{X_\xi}-H_{X_\xi^\cp}=R'-H_{X_\xi^\cp}.\]  Using these definitions and that $[H_{X_\xi^\cp},A^\xi_s]=0$
	for all $s\in\Rb,\xi\ge0,$ since $\supp A^\xi_s\subset X_\xi$ (see \eqref{Axi}),
	%
	we find from \eqref{410} that
	\begin{equation}\label{4.41}
		\begin{aligned}
			\Rem_t 
			= & \int_0^t \al_r\del{i\sbr{R,A^\xi_{t-r}}}\,dr.
		\end{aligned}
	\end{equation}
 	
 	Next, we use the standing assumptions $h_{xy}=\overline {h_{yx}}$ and $v_{xy}=v_{yx}=\overline{ v_{xy}}$ to split the ($X_\xi$-$X_\xi^c$)-coupling term $R$ 
  into two terms arising respectively from the kinetic and potential terms in \eqref{1.1} 
   (see \figref{fig:splitting}): 
\begin{align}
	\label{5.12}
	R:=& S+W, \\
	\label{5.10}
	S:=&S'+(S')^*,\quad S':= \sum_{x\in X_\xi,y\in X_\xi^\cp}h_{xy}a_x^*a_y,\\
	W:=&\sum_{x\in X_\xi,y\in X_\xi^\cp}V_{xy}\quad \text{ with $V_{xy}:=a_x^*a_y^*v_{xy}a_ya_x$}.\label{W-def}
\end{align} 
Eqs.~\eqref{4.41}--\eqref{5.12} imply, for any $\varphi,\,\psi\in\cD(N)\cap \cD(H)$,
	\begin{multline}\label{4.42}
		\abs{\undel{\Rem_t}}\le   t\sup_{0\le r \le t}\del{\abs{\urdel{\sbr{S,A^\xi_{t-r}}}}+\abs{\urdel{\sbr{W,A^\xi_{t-r}}}}},
	\end{multline}
where $\phi_r=e^{-irH}\phi$. 
In the rest of the proof we estimate the r.h.s.~of this expression.

	2. We first estimate the first term in the r.h.s.~of \eqref{4.42}. {Within this step, all constants $C>0$ depend only on $n,\,\kappa_n,$ and $c$.}

	Let $s:=t-r$. By formula \eqref{5.10},
	we have the estimate
	\begin{align*}
		&\abs{\urdel{\sbr{S',A^\xi_{s}}}}\\\le& \sum_{x\in X_\xi,y\in X_\xi^\cp} \abs{h_{xy}}\del{\abs{\urdel{a_x^*a_yA^\xi_{s}}}+\abs{\urdel{A^\xi_{s}a_x^*a_y}}}\\
		=& \sum_{x\in X_\xi,y\in X_\xi^\cp} \abs{h_{xy}}\del{\abs{\inn{a_x\varphi_r}{a_yA^\xi_{s} \psi_r}}+\abs{\inn{a_x(A^\xi_{s})^*\varphi_r}{a_y \psi_r}}}.
	\end{align*}
	Using the Cauchy-Schwarz inequality and the fact that $\norm{a_z\varphi_r}=\ordel{n_z}^{1/2}$ (recall $n_z=a_z^*a_z$), we find
	\begin{align}
		&\abs{\urdel{\sbr{S',A^\xi_{s}}}}
		\le \mathrm{I}^{1/2}\mathrm{II}^{1/2}+\mathrm{III}^{1/2}\mathrm{IV}^{1/2},\label{4.71}\\
		&\mathrm{I}:=\sum_{x\in X_\xi,y\in X_\xi^\cp} \abs{h_{xy}}\ordel{n_x},\label{4.71.1}\\
		&\mathrm{II}:=\sum_{x\in X_\xi,y\in X_\xi^\cp} \abs{h_{xy}}\wrdel {(A^\xi_{s})^*n_yA^\xi_{s} },\label{4.71.2}\\
		&\mathrm{III} :=\sum_{x\in X_\xi,y\in X_\xi^\cp} \abs{h_{xy}}\wrdel{n_y},\label{4.71.3}\\
		&\mathrm{IV} :=\sum_{x\in X_\xi,y\in X_\xi^\cp} \abs{h_{xy}}\ordel {A^\xi_{s}  n_x (A^\xi_{s})^*}  \label{4.71.4}.
	\end{align}
We now estimate the terms in the r.h.s.~of \eqref{4.71}. Given $c>2\kappa$, we fix a number $0<\g<1/3$ in the remainder of the proof s.th.
		\begin{equation}\label{gamma-cond}
			c_1:=\frac{(1-\g)c}{2}>\kappa.
		\end{equation}
We introduce the local number operators  counting the number of particles in the curved annular regions (c.f.~\eqref{55s} and \figref{fig:N} below):  \begin{equation}\label{N1def}
		N_{\g,\xi}:=N'_{\g,\xi}+N''_{\g,\xi},\quad	N'_{\g,\xi}  := N_{ X_{(1-\g)\xi,\xi}},\quad 	N''_{\g,\xi}  := N_{ X_{\xi,(1+\g)\xi}}.
	\end{equation} 
	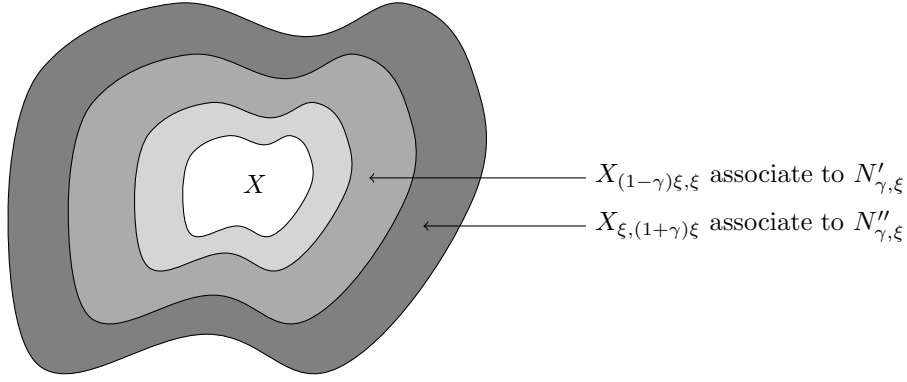
\begin{figure}[H]
		\centering
		\begin{tikzpicture}[scale=.8]
			\draw[fill=gray!100]  plot[shift={(-0.3,-0.45)}, scale=1.1,smooth, tension=.7] coordinates {(-3,0.5) (-2.5,2.5) (-.5,3.5) (1.5,3) (3,3.5) (4,2.5) (4,0.5) (2,-2) (0,-1.5) (-2.5,-2) (-3,0.5)};
			
			\draw[fill=gray!66]  plot[shift={(-0.2,-0.25)}, scale=.8,smooth, tension=.7] coordinates {(-3,0.5) (-2.5,2.5) (-.5,3.5) (1.5,3) (3,3.5) (4,2.5) (4,0.5) (2,-2) (0,-1.5) (-2.5,-2) (-3,0.5)};
			
			\draw[fill=gray!33]  plot[scale=.5,smooth, tension=.7] coordinates {(-3,0.5) (-2.5,2.5) (-.5,3.5) (1.5,3) (3,3.5) (4,2.5) (4,0.5) (2,-2) (0,-1.5) (-2.5,-2) (-3,0.5)};
			
			\draw[fill=white]  plot[shift={(0.2,.15)},scale=.3,smooth, tension=.7] coordinates {(-3,0.5) (-2.5,2.5) (-.5,3.5) (1.5,3) (3,3.5) (4,2.5) (4,0.5) (2,-2) (0,-1.5) (-2.5,-2) (-3,0.5)};
			
			\node at (.5,.4) {$X$};

			\draw [->] (6,.5)--(2.4,.5);
			\node [right] at (6,.5) {$X_{(1-\g)\xi,\xi}\text{ associate to }N'_{\g,\xi}$};
			
			\draw [->] (6,-.3)--(3.3,-.3);
			\node [right] at (6,-0.3) {$X_{\xi,(1+\g)\xi}\text{ associate to }N''_{\g,\xi}$};
		\end{tikzpicture}
		\caption{Schematic diagram illustrating the region associated to the local number operators in \eqref{N1def}.}
		\label{fig:N}
	\end{figure}

		
In what follows, we will use the following estimate for the particle number in curved annular regions (c.f.~\cite[Thm.~2.3]{FLS2}):
\begin{align}\label{624}
 	\ordel{N_{\g_1,\xi}}
	\le C\del{\ondel{N_{\g_2,\xi}}+((\g_2-\g_1)\xi)^{-n} \ondel{N}},
\end{align}
valid for any two numbers $1\ge \g_2>\g_1\ge0$, $c>\kappa$, and  $r<(\g_2-\g_1)\xi/c$. Estimate \eqref{624} follows from  the `incoming' light cone estimate,  \eqref{MVEext}.

	2.1.  To estimate the term $\mathrm{I}$ from \eqref{4.71.1},
	we use the decomposition
	\begin{equation}\label{X-decomp}
		X_\xi=X_{(1-\g)\xi}\cup  X_{(1-\g)\xi, \xi},
	\end{equation}
	c.f.~\figref{fig:N}, and that $\abs{x-y}\ge\g\xi$ for all $x\in X_{(1-\g)\xi}$ and $y\in X_\xi^\cp$, to compute
	\begin{align}
		\mathrm{I}
		&=\sum_{y\in X_\xi^\cp} |h_{xy}|\del{\sum_{x\in  X_{(1-\g)\xi,\xi}} \ordel{n_x}  +\sum_{x\in X_{(1-\g)\xi}} \ordel{n_x}}   \notag \\ 
		&\le \del{\sup_{x\in\Lam}\sum_{y\in X_\xi^\cp} |h_{xy}|}\sum_{x\in  X_{(1-\g)\xi,\xi}} \ordel{n_x} \notag \\&\quad +(\g\xi)^{-n} \del{\sup_{x\in\Lam}\sum_{y\in X_\xi^\cp} |h_{xy}||x-y|^n}\sum_{x\in X_{(1-\g)\xi}} \ordel{n_x}. \notag 
	\end{align}
	Recalling definition \eqref{k-cond} and noting the fact that $\sup_{x\in\Lam}\sum_{y\in X_\xi^\cp} |h_{xy}|\le\kappa_{n-1}$ as the grid size of the underlying lattice is at least $1$ (see \eqref{k-cond}), we conclude
	\begin{equation}
		\mathrm{I} \le\kappa_{n-1}\del{ \ordel{N'_{\g,\xi} }+(\g\xi)^{-n }\ordel{N} }, \label{4.811}
	\end{equation}
	where, recall, $N_{\g,\xi} '\equiv N_{ X_{(1-\g)\xi,\xi}}$ (see \eqref{N1def}).

To estimate the first term in line \eqref{4.811},  we use the relation $N_{\g,\xi}'\le N_{\g,\xi}$ and apply
\eqref{624} with  {$c\to c_1>\kappa$ (see the choice \eqref{gamma-cond})}  to obtain, {for all $0\le r< (1-\g)\xi/c_1$,}
{	\begin{align}\label{8.909}
			\ordel{N_{\g,\xi}'}\le&C\tau_0(\varphi),
				\\\tau_0(\phi):=&{\phidel{  N_{1,\xi} }+(\g\xi)^{-n} \phidel{N}}.\label{tau0def}
	\end{align}}
Plugging estimate \eqref{8.909} back to \eqref{4.811} and using the conservation $N$, we find that 
	\begin{equation}\label{4.804}
		\mathrm{I}\le C\tau_0(\varphi),
	\end{equation}
	uniformly for all $r$ with  {$0\le r< (1-\g)\xi/c_1$}.
	
	2.2. To estimate the term $\mathrm{II}$ from \eqref{4.71.2}, we use the decomposition \begin{equation}\label{Xc-decomp}
		X_\xi^\cp=X_{(1+\g)\xi}^\cp\cup  X_{\xi, (1+\g)\xi},
	\end{equation}c.f.~\figref{fig:N}, and the notation $\widetilde {n_y}:=(A^\xi_{s})^*n_yA^\xi_{s}$, to compute
	\begin{align}
		\mathrm{II}=&\sum_{x\in X_\xi} |h_{xy}|\del{\sum_{y\in  X_{\xi,(1+\g)\xi}} \wrdel {\widetilde {n_y} }+\sum_{y\in X_{(1+\g)\xi}^\cp} \wrdel {\widetilde {n_y} } }  \notag \\ 
		\le	&\del{\sup_{y\in\Lam}\sum_{x\in X_\xi} |h_{xy}|}\sum_{y\in  X_{\xi,(1+\g)\xi}}\wrdel {\widetilde {n_y} }   \notag\\&\quad+(\g\xi)^{-n} \del{\sup_{y\in\Lam}\sum_{x\in X_\xi} |h_{xy}||x-y|^n}\sum_{y\in X_{(1+\g)\xi}^\cp} \wrdel {\widetilde {n_y} }   \notag \\ 
		\le& \kappa_{n-1}\del{ \wrdel{(A^\xi_{s})^*N_{\g,\xi}''A^\xi_{s} }+(\g\xi)^{-n }\wrdel{(A^\xi_{s})^* NA^\xi_{s}} }. \label{4.812}
	\end{align}
To estimate the first term in line \eqref{4.812}, we note that since $\supp A^\xi_s\subset X_\xi$ for all $s$ and $\supp N''_{\g,\xi}\subset X_\xi^\cp$ by construction (see \eqref{N1def}), we have $\sbr{A^\xi_s,N''_{\g,\xi}}\equiv 0$ for all $s$. By this fact, the first term in line \eqref{4.812} can be bounded as 
	\begin{equation}\label{627}
\begin{aligned}
			&\wrdel{(A^\xi_{s})^*N_{\g,\xi}''A^\xi_{s} }\\=&\wrdel{(N_{\g,\xi}'')^{1/2}(A^\xi_{s})^*A^\xi_{s} (N_{\g,\xi}'')^{1/2}}\\\le& \Norm{A^\xi_{s}}^2\wrdel{N_{\g,\xi}'' }=\Norm{A}^2\wrdel{N_{\g,\xi}'' }.
\end{aligned}
	\end{equation}
	Using the relation $N_{\g,\xi}''\le N_{\g,\xi}$ and applying \corref{cor3} to the r.h.s. above with {$c\to c_1>\kappa$}, we  obtain  that for all {$0\le r < (1-\g) \xi/c_1$},
	\begin{align}\label{628}
			&\wrdel{(A^\xi_{s})^*N_{\g,\xi}''A^\xi_{s} }\le C\Norm  {A}^2\tau_0(\psi).
	\end{align}
To estimate the second term in line \eqref{4.812}, we use the relation $[A_s^\xi,N]=0$ (see \eqref{Axi}) and the conservation of the expectation of $N$ to get 
	$$\wrdel{(A^\xi_{s})^* NA^\xi_{s}}\le \Norm{A^\xi_{s}}^2\wrdel{N}=\Norm{A}^2\wndel{N}.$$
	Plugging the two preceding inequalities back to \eqref{4.812},  we obtain
	\begin{equation}\label{4.805}
		\begin{aligned}
			&\mathrm{II}\le C\Norm{A}^2 \tau_0(\psi).
		\end{aligned}
	\end{equation}
	uniformly for all $r$ with {$0\le r < (1-\g) \xi/c_1$}.

	2.3. The term $\mathrm{III}$ in \eqref{4.71.3}  can be bounded as \eqref{4.812}. (It is actually simpler because there is no $A$'s in \eqref{4.71.3}.) Here we record the result:
	\begin{equation}\label{4.801'}
		\mathrm{III} \le C\tau_0(\psi),
	\end{equation}
	which holds {uniformly for all $0\le r< (1-\g)\xi/c_1$.}

	2.4. To bound the term $\mathrm{IV}$ in \eqref{4.71.4}, we use the decomposition \eqref{X-decomp}
	to compute 
	\begin{align}
		\mathrm{IV} =&\sum_{y\in X_\xi^\cp} |h_{xy}|\sum_{x\in X_{(1-\g)\xi}}\ordel {A^\xi_{s}  n_x (A^\xi_{s})^*} \label{4.802} \\
		& + \sum_{y\in X_\xi^\cp} |h_{xy}|\sum_{x\in  X_{(1-\g)\xi,\xi}}\ordel {A^\xi_{s}  n_x (A^\xi_{s})^*}  \label{4.803}.
	\end{align}
Using the relation $[A_s^\xi,N]=0$ (see \eqref{Axi}) and the conservation of $N$, we bound the term in line \eqref{4.802} as 
	\begin{align}
		&\quad\sum_{y\in X_\xi^\cp} |h_{xy}|\sum_{x\in X_{(1-\g)\xi}}\ordel {A^\xi_{s}  n_x(A^\xi_{s})^*} \notag \\ 
		&\le (\g\xi)^{-n}  \sum_{y\in X_\xi^\cp} |h_{xy}||x-y|^n\sum_{x\in X_{(1-\g)\xi}} \ordel{A^\xi_{s}  n_x(A^\xi_{s})^*}\notag \\ 
		&\le\kappa_{n-1}(\g\xi)^{-n } \ordel{A^\xi_{s}N(A^\xi_{s})^*}\notag\\
		&\le \kappa_{n-1}\Norm{A}^2(\g\xi)^{-n}\ondel{N}.\label{4.802'}
	\end{align}
	To bound the term in line \eqref{4.803}, we 
		define $ \varphi_{r,s}:=e^{-isH_{X_\xi}} \varphi_r$ and recall $N_{\g,\xi} '\equiv N_{ X_{(1-\g)\xi,\xi}}$.
		Then, we have, by definition \eqref{Axi} for the local evolution,
		\begin{equation}\label{4.831'}
			\begin{aligned}
				&\sum_{y\in X_\xi^\cp} |h_{xy}|\sum_{x\in  X_{(1-\g)\xi,\xi}}\ordel {A^\xi_{s}  n_x (A^\xi_{s})^*} 
				\\\le& \kappa_{n-1} \ordel {A^\xi_{s} N_{\g,\xi} '(A^\xi_{s})^*} 
				\\=&\kappa_{n-1}\inn{ A^* \varphi_{r,s}} {  \al^{X_\xi}_{-s}( N_{\g,\xi}')A^* \varphi_{r,s}}.
			\end{aligned}
		\end{equation}
To estimate the quantity in the last line of \eqref{4.831'}, we 
		use \corref{cor3} with evolution $\al_{-s}^{X_\xi}(\cdot)$ to obtain  that for all {$s <  \tfrac{1-\g}{2}\xi/c_1$},
		\begin{equation}
			\label{4.831}
			\begin{aligned}
				&\inn{ A^* \varphi_{r,s}} {  \al^{X_\xi}_{-s}( N_{\g,\xi}')A^* \varphi_{r,s}}\\
				\le&C\del{\inn{ \varphi_{r,s}}{ A N_{(1+\g)/2,\xi}' A^* \varphi_{r,s} }+(\g\xi)^{-n}\inn{ \varphi_{r,s}}{ A  N A^* \varphi_{r,s}}}.
			\end{aligned}
		\end{equation}
		For the remainder estimate, we use that $0<\g<1/3$ so that $\tfrac{1-\g}{2}>\g$. 
{		Note that this is the only place $\g<1/3$ is used. }
		
		Since $\supp A\subset X$ and $\supp N_{(1+\g)/2,\xi}'\subset X^\cp$ by construction, we can pull out $A$'s from \eqref{4.831} to obtain 
		\begin{equation}
			\label{4.832}
			\begin{aligned}
				&\inn{ A^* \varphi_{r,s}} {  \al^{X_\xi}_{-s}( N_{\g,\xi}')A^* \varphi_{r,s}}\\
				\le&\Norm{A}^2C\del{\inn{ \varphi_{r,s}}{  N_{(1+\g)/2,\xi}'   \varphi_{r,s} }+(\g\xi)^{-n}\inn{ \varphi_{r,s}}{  N \varphi_{r,s}}}.
			\end{aligned}
		\end{equation}
Now we use \corref{cor3} twice on the first term of \eqref{4.832}, first with the evolution $\al^{X_\xi}_s(\cdot)$ and then with $\al_r(\cdot)$. This way we obtain that for {$r+s=t< \tfrac{1-\g}{2}\xi/c_1$},
		\begin{align}\label{4.833}
				\inn{ \varphi_{r,s}}{  N_{(1+\g)/2,\xi}'   \varphi_{r,s} } \le &\inn{ \varphi}{ \al_r \circ \al_s^{X_\xi}( N_{(1+\g)/2,\xi}' ) \varphi }\le C\tau_0(\varphi),
		\end{align}
	{where $\tau_0(\varphi)$ is defined by \eqref{tau0def}.}
		This bounds the first term in the r.h.s.~of \eqref{4.832}. 	(The derivation is similar to \eqref{8.5}.) 
		By the conservation of $N$, we find 
		\begin{equation}\label{4.834}
			\inn{ \varphi_{r,s}}{  N \varphi_{r,s}}=\ondel{N}
		\end{equation}
		in the second term in the r.h.s. of \eqref{4.832}.

		Combining \eqref{4.802'}--\eqref{4.834} yields
		\begin{equation}\label{4.803'}
			\begin{aligned}
				\mathrm{IV}\le & C\Norm{A}^2 \tau_0(\varphi),
			\end{aligned}
		\end{equation}
		which holds uniformly for all  
		\begin{equation}\label{s-cond}
			{t< \frac{1-\g}{2}\xi/c_1. }
		\end{equation}
		
		2.5. At this point, we have uniform estimates \eqref{4.804}, \eqref{4.805}, \eqref{4.801'}, and \eqref{4.803'}, which are valid for all $t$ satisfying \eqref{s-cond}. Plugging these estimates back to \eqref{4.71}, we conclude that all $t$ satisfying \eqref{s-cond},
		\begin{align*}
			\sup_{0\le r \le t}\abs{\urdel{\sbr{S',A^\xi_{t-r}}}}\le C \Norm{A}\tau_0(\varphi)^{1/2}\tau_0(\psi)^{1/2}.
		\end{align*}
		Since $\abs{\urdel{\sbr{(S')^*,A^\xi_{t-r}}}}=\abs{\inn{\psi_r}{\sbr{S',(A^\xi_{t-r})^*}\varphi_r}}$ and $(A^\xi_{t-r})^*=(A^*)^\xi_{t-r}$ (see \eqref{Axi}), going through Steps 2.1--4 and interchanging the roles of $A^*$ (resp. $\varphi_r$) and $A$ (resp. $\psi_r$) yields the exact same estimate for $\sup_{0\le r \le t}\abs{\urdel{\sbr{(S')^*,A^\xi_{t-r}}}}$ as above.

		Recalling the definition of $c_1$ in \eqref{gamma-cond} and the validity interval \eqref{s-cond},  we conclude  that {for every $0\le t< \xi/c$,}
		\begin{align}\label{4.84}
			\sup_{0\le r \le t}\abs{\urdel{\sbr{S,A^\xi_{t-r}}}}\le C\Norm{A}\tau_0(\varphi)^{1/2}\tau_0(\psi)^{1/2}.
		\end{align}
		This bounds the first term in the r.h.s.~of \eqref{4.42}.

		3.
		Next,we estimate the second  term in the last line of \eqref{4.42}. {Within this step, all constants $C>0$ depend only on $n,\,\kappa_n,\,\nu_n,$ and $c$, where $\nu_n$ is as in \eqref{v-cond}. }
		
		By formula \eqref{W-def}, the fact that $[a_x^\sharp,n_y]\equiv 0$ for all $x\in X_\xi,\,y\in X_\xi^\cp$ and $a_x^\sharp= a_x,\,a_x^*$, and the localization property $[A^\xi_s,a_y]=0$,
		we have the estimate
		\begin{align*}
			&\abs{\urdel{\sbr{W,A^\xi_{s}}}}\\\le& \sum_{x\in X_\xi,y\in X_\xi^\cp} \abs{v_{xy}}\del{\abs{\inn{\varphi_r}{n_xn_yA^\xi_{s} \psi_r}}+\abs{\inn{\varphi_r}{A^\xi_{s}n_xn_y \psi_r}}}\\
			=&\sum_{x\in X_\xi,y\in X_\xi^\cp} \abs{v_{xy}}\del{\abs{\inn{(A^\xi_s)^*n_x\varphi_r}{n_y \psi_r}}+\abs{\inn{n_x (A^\xi_s)^*\varphi_r}{n_y \psi_r}}}.
		\end{align*}
		Applying the Cauchy-Schwarz and triangle inequalities to the last line above, we find
		\begin{align}
			&\abs{\urdel{\sbr{W,A^\xi_{s}}}}\le \mathrm{V}^{1/2}\mathrm{VI}^{1/2}+\mathrm{V}^{1/2}\mathrm{VII}^{1/2}, \label{4.60}\\
			& \mathrm{V}:=\sum_{x\in X_\xi,y\in X_\xi^\cp} \abs{v_{xy}}\wrdel{n_y^2}\label{4.601},\\
			&\mathrm{VI}:= \sum_{x\in X_\xi,y\in X_\xi^\cp} \abs{v_{xy}}\ordel{n_xA^\xi_{s}(A^\xi_{s})^*n_x},\label{4.602}\\
			&\mathrm{VII}:= \sum_{x\in X_\xi,y\in X_\xi^\cp} \abs{v_{xy}}\ordel{A^\xi_{s}n_x^2(A^\xi_{s})^*}.\label{4.603}
		\end{align}
	
		
		3.1. To bound the term $\mathrm{V}$ in \eqref{4.601}, we use the fact $\sum_{x\in S} n_x^2\le N_S^2$ and proceed exactly as in Step 2.3 above for the estimate of \eqref{4.71.3} (see more details in Step 3.2 below). This way we obtain that, {for all $0\le r<(1-\g)\xi$, }
		\begin{align}\label{4.61}
				\mathrm{V}\le& C\tau(\psi),
\\\label{taudef}
\tau(\phi):=&\phidel{N_{1,\xi}N}+(\g\xi)^{-n} \phidel{N^{2}}.
		\end{align}

		3.2. To bound the term $\mathrm{VI}$ in \eqref{4.602}, we first use 
		$$
		\mathrm{VI}\le \Norm{A_t^\xi}^2 \sum_{x\in X_\xi,y\in X_\xi^\cp} \abs{v_{xy}}\ordel{n_x^2}=\Norm{A}^2 \sum_{x\in X_\xi,y\in X_\xi^\cp} \abs{v_{xy}}\ordel{n_x^2}.
		$$
		Using decomposition \eqref{X-decomp} and the relation $\sum_{x\in S} n_x^2\le N_S^2$,  we then proceed as in the estimate of \eqref{4.71.1} (see Step 2.1) to obtain, in place of \eqref{4.811},
		\begin{align*}
			\mathrm{VI} &\le \Norm{A}^2\Bigl( \sum_{y\in X_\xi^\cp} |v_{xy}|\sum_{x\in  X_{(1-\g)\xi,\xi}} \ordel{n_x^2}  \\&\quad +(\g\xi)^{-n} \sum_{y\in X_\xi^\cp} |v_{xy}||x-y|^n\sum_{x\in X_{(1-\g)\xi}} \ordel{n_x^2}\Bigr)  \\ 
			&\le\nu_n\del{ \ordel{ \del{N'_{\g,\xi}}^2  }+(\g\xi)^{-n }\ordel{N^2} }.
		\end{align*}
		For the first term in the last line above, we use \corref{cor3} to obtain, for all {$0\le r < (1-\g)\xi/c_1$,}
		\begin{equation}\label{647}
			\ordel{ \del{N'_{\g,\xi}}^2  }\le 
			C\tau(\varphi).
		\end{equation}
		Using the preceding two estimates and the conservation of $N$, we conclude that 
		\begin{equation}\label{4.62}
			\begin{aligned}
				\mathrm{VI}\le C\Norm{A}^2\tau(\varphi),
			\end{aligned}
		\end{equation}
		uniformly for all {$0\le r < (1-\g)\xi$.}
		
		3.3. 
		To bound the term $\mathrm{VII}$ in \eqref{4.603}, we  proceed as in the estimate of \eqref{4.71.4} (see Step 2.4). Using the decomposition \eqref{X-decomp}, we compute
		\begin{align}
			\mathrm{VII} =&\sum_{y\in X_\xi^\cp} |v_{xy}|\sum_{x\in X_{(1-\g)\xi}}\ordel {A^\xi_{s}  n_x^2 (A^\xi_{s})^*} \label{4.901} \\
			& + \sum_{y\in X_\xi^\cp} |v_{xy}|\sum_{x\in  X_{(1-\g)\xi,\xi}}\ordel {A^\xi_{s}  n_x^2 (A^\xi_{s})^*}  \label{4.902}.
		\end{align}
The term in line \eqref{4.901} can be bounded in the same way as \eqref{4.802'}:
		\begin{align}
			\quad\sum_{y\in X_\xi^\cp} |v_{xy}|\sum_{x\in X_{(1-\g)\xi}}\ordel {A^\xi_{s}  n_x^2(A^\xi_{s})^*} 
			&\le \nu_n(\g\xi)^{-n}\Norm{A}^2\ondel{N^{2}}.\label{4.901'}
		\end{align}
		To bound the term in line \eqref{4.902}, we first use that \begin{align*}
			\sum_{y\in X_\xi^\cp} |v_{xy}|\sum_{x\in  X_{(1-\g)\xi,\xi}}\ordel {A^\xi_{s}  n_x^2 (A^\xi_{s})^*} 
			\le&\nu_n\inn{ A^* \varphi_{r,s}} {  \al^{X_\xi}_{-s}( (N_{\g,\xi}')^2)A^* \varphi_{r,s}},
		\end{align*}
		which can be derived similarly as \eqref{4.831'}. Applying \corref{cor3} to the r.h.s. above, we find that for all $s<\tfrac{1-\g}{2}\xi/c_1$, 
		\begin{equation}
			\label{4.903}
			\begin{aligned}
				&\inn{ A^* \varphi_{r,s}} {  \al^{X_\xi}_{-s}( (N_{\g,\xi}')^2)A^* \varphi_{r,s}}\\
				\le&C\del{\inn{ \varphi_{r,s}}{ A N_{(1+\g)/2,\xi}' NA^* \varphi_{r,s} }+(\g\xi)^{-n}\inn{ \varphi_{r,s}}{ A  N^2 A^* \varphi_{r,s}}}.
			\end{aligned}
		\end{equation}
		{Since the local number operator $N_{(1+\g)/2,\xi}'$ is supported away from $X \supset \supp A$ (see \eqref{N1def} and \figref{fig:N}),} and since $[A,N] =0$, we can pull out the $A$'s from the first term in the r.h.s. of \eqref{4.903} as
		\begin{align}\label{theplace}
			\inn{ \varphi_{r,s}}{ A N_{(1+\g)/2,\xi}' NA^* \varphi_{r,s} }\le& \Norm{A}^2\inn{ \varphi_{r,s}}{  N_{(1+\g)/2,\xi}'N \varphi_{r,s} },
		\end{align}
	c.f.~\eqref{627}.
		The rest of the estimate of \eqref{4.903} follows similarly as in the estimate of \eqref{4.831}.
		Here we record the result: \begin{equation}\label{4.63}
			\begin{aligned}
				\mathrm{VII}\le C\Norm{A}^2\tau(\varphi),
			\end{aligned}
		\end{equation}
		which holds uniformly for all $0\le r< t,\,s=t-r$, so long as \eqref{s-cond} holds.
		
		3.4. Combining uniform estimates \eqref{4.60}, \eqref{4.61}--\eqref{4.63} yields, {for every $0\le t< \xi/c$,}
		\begin{align}
			\label{4.64}
			\sup_{0\le r \le t}	\abs{\urdel{\sbr{W,A^\xi_{t-r}}}} \le C \Norm{A}\tau(\varphi)^{1/2}\tau(\psi)^{1/2}.
	\end{align}

		4. Combining \eqref{4.84} and \eqref{4.64} in \eqref{4.42} {and recalling the choice $\g=\g(c,\kappa)$ in \eqref{gamma-cond}} and definitions \eqref{tau0def}, \eqref{taudef}, we conclude that, for $t<\eta/c$ and $c>2\kappa$,
		\begin{align}
				\abs{\undel{\Rem_t}}\le  C t  \Norm{A}\tau(\varphi)^{1/2}\tau(\psi)^{1/2}.\label{RemFormEst}
		\end{align}
	Finally, for any mixed state $\om$ satisfying \eqref{emptyshellcond} and any operator $B\in\cB_Y$, we use the spectral decomposition $\om=\sum p_i P_{\psi^i}$ with $p_i\ge0,\,\sum p_i\le C<\infty$, and the choice $\varphi^i=B^*\psi^i$ in \eqref{RemFormEst} to obtain
	\begin{align}\label{lcae-gen-pf1}
		|\om\big(B\,\Rem_t\big)|\le& \sum p_i\nonumber \abs{\inn{B^*\psi^i }{\Rem_t\psi^i}} \\\le& Ct \Norm{A}\sum p_i\, \tau(B^*\psi^i)^{1/2}\tau(\psi^i)^{1/2}.
	\end{align}
Since $B\in\cB_Y$ and $Y\subset X_{2\xi}^\cp$, we have $[B,N]=[B,N_{1,\xi}]=0$ (see \eqref{N1def}). Therefore, by definition \eqref{taudef}, we have $\tau(B^*\psi^i) \le \norm{B}^2\tau(\psi^i)$ for each $i$.
This, together with estimate \eqref{lcae-gen-pf1} and the facts that $\sum p_i\tau(\psi_i)=\om\del{N_{X_{2\xi}\setminus X}N}+(\g\xi)^{-n}\om(N^2)$ (see \eqref{taudef}) and $\,0<\g<1/3$, yields
\begin{align}
			|\om\big(B\,\Rem_t\big)|\le& C\g^{-n}t \Norm{A}\Norm{B} \del{ \om\del{N_{X_{2\xi}\setminus X}N}+\xi^{-n}\om(N^2)}.\end{align}
This completes the proof of Theorem \ref{thm2-gen}. 	\qed
		
		\section{Generalizations to unbounded observables}\label{sec:nu}
		
		\renewcommand{\Norm}[1]{{\left\vert\kern-0.25ex\left\vert\kern-0.25ex\left\vert #1 
				\right\vert\kern-0.25ex\right\vert\kern-0.25ex\right\vert}}

		In this section, we sketch the extension of the main theorems in \secref{sec:setup} to a large class of unbounded operators. We say that an operator $A$ acting on $\cF$ \textit{has finite degree} if $[A,N]=0$ and 
		\begin{equation}\label{1.18}
			\deg A:=\inf\Set{\nu\ge0:	A,\,A^*\text{ are $N^{\nu/2}$-bounded}}<\infty.
		\end{equation}
		By definition, $\deg A=0$ if and only if $A$ is bounded, and $\deg A\le 2M$ if $A$ is a polynomial in $n_x$ with degree at most $M$. For each $0\le\nu<\infty$, we define the norm
		\begin{equation}\label{5.33}
			\Norm{A}_\nu:= \max\del{\norm{AN^{-\nu/2}},\norm{A^*N^{-\nu/2}}}.
		\end{equation}
		Let \begin{align}
			&\cB^\nu:= \Set{\text{operators $A$ on $\cF$ with  $\Norm{A}_\nu<\infty$}},\\
			&\cB^\nu_X:=\Set{A\in\cB^\nu:\supp A\subset X,\,[A,N]=0}.\label{Bnudef}
		\end{align}
		Then $\cB^\nu_X$ with $\nu=0$ coincides with \eqref{BX}, and we have the following lemma: 
		\begin{lemma}
			\label{lemA.3}
			Let $0\le\nu<\infty$, $X\subset\Lam$, and $A\in\cB^\nu_X$. Then, for all numbers $p,q\ge0$ and operator $B\ge0$ with $[B,N]=0,\,\supp B\subset X^\cp$, we have  the following operator inequalities:
			\begin{align}\label{A5.34}
				A^*N^qA\le& \Norm{A}_\nu^2N^{\nu+q}&\text{ on $\cD(N^{(\nu+q)/2})$},\\
				A^*N^{p/2}B^qN^{p/2}A\le& \Norm{A}_\nu^2N^{(\nu+p)/2}B^q N^{(\nu+p)/2}&\text{ on $\cD(B^{q/2}N^{(\nu+p)/2})$}.\label{A5.34'}
			\end{align}
		\end{lemma}
		\begin{proof}
			Since $\nu$ is fixed, we write $\Norm{\cdot}\equiv \Norm{\cdot}_\nu$ within this proof. Symmetrizing as $$ A^* N ^qA=N ^{(\nu+q)/2} N ^{-(\nu+q)/2} A^* N ^{q/2} N ^{q/2} A N ^{-(\nu+q)/2} N ^{(\nu+q)/2},$$
			and using definitions \eqref{5.33}--\eqref{Bnudef},
			we see that $$
			A^*N^qA\le\norm{N ^{q/2} A N ^{-(\nu+q)/2}}^2N^{\nu+q}= \Norm{A}^2N^{\nu+q}.
			$$
			This gives \eqref{A5.34}. Next, since $\supp A\in X$ and $\supp B\subset X^\cp$, we have $[A,B]=[A,N]=0$ and therefore
			$$A^*N^{p/2}B^qN^{p/2}A=N^{p/2}B^{q/2}A^*AB^{q/2}N^{p/2}.$$ Applying \eqref{A5.34} to the r.h.s. with $q=0$ and then using the fact that $[N,B]=0$, we find
			\begin{align*}	
				A^*N^{p/2}B^qN^{p/2}A\le& \Norm{A}^2N^{p/2}B^{q/2}N^\nu B^{q/2}N^{p/2}\\=&\Norm{A}^2N^{(\nu+p)/2}B^q N^{(\nu+p)/2}.
			\end{align*}
			This gives \eqref{A5.34'}.
		\end{proof}
		
		Now we sketch the proof of \thmref{thm2-gen} for operators $A\in\cB^\nu_X$ with $\nu>0$. Through this, the corresponding results for Thms.~\ref{thm2}--\ref{thm3} follow readily, wherefore extensions of the results in Sects.~\ref{sec:cor}--\ref{sec:2.8}, which are applications of Thms.~\ref{thm2}--\ref{thm3}, follow.
		
		The main idea is to use \corref{cor3} together with \lemref{lemA.3}, in places where \thmref{thm1} is used. For example, in place of \eqref{627}, we use \eqref{A5.34} with $B=N''_{\g,\xi}$, which is supported away from $X_\xi$ (see \eqref{N1def}), to get
		\begin{equation}\label{627'}
			\wrdel{(A^\xi_{s})^*N_{\g,\xi}''A^\xi_{s} }\le \Norm{A}^2\wrdel{N_{\g,\xi}''N^{\nu} }.
		\end{equation}
		Then we use \corref{cor3} to the r.h.s. above to get
		$$
		\begin{aligned}
			\wrdel{N_{\g,\xi}''N^{\nu} }
			\le& C  \del{\wndel{N_{1,\xi}N^{\nu}}+(\g\xi)^{-n} \wndel{N^{\nu+1}}}.
		\end{aligned}
		$$
		This, together with \eqref{627'}, yields \begin{equation}\label{628'}
			\begin{aligned}
				&\wrdel{(A^\xi_{s})^*N_{\g,\xi}''A^\xi_{s} }
				\\\le& C\Norm  {A}^2\del{\wndel{N_{1,\xi}N^{\nu}}+(\g\xi)^{-n} \wndel{N^{\nu+1}}}
			\end{aligned}
		\end{equation}
		in place of \eqref{628}. Similar modifications are then made to \eqref{4.832}, \eqref{4.833}, \eqref{647}, etc. This way one can obtain,  for  states $\varphi,\,\psi\in\cD(N^{\frac{\nu+2}{2}})\cap\cD(H)$  and all other notations the same as in \thmref{thm2-gen},
		$$
		\begin{aligned}
			&\abs{\undel{(A_t-A_t^\xi)}}\\
			\le& C\abs{t}\Norm{A}
			\del{\ondel{N_{X_{2\xi}\setminus X}N^{\nu+1}}+(\g\xi)^{-n} \ondel{N^{\nu+2}}}^{1/2}
			\\
			&\quad \times  \del{\wndel{N_{X_{2\xi}\setminus X}N^{\nu+1}}+ (\g\xi)^{-n}\wndel{N^{\nu+2}}}^{1/2}.
		\end{aligned}
		$$

	\end{document}